\documentclass[10pt]{article}
\usepackage[preprint]{tmlr}

\usepackage{amsmath,amssymb,amsthm,bm}

\usepackage{booktabs,longtable,array,caption,subcaption,multirow}
\usepackage{graphicx,float}
\usepackage{adjustbox}
\usepackage{makecell}
\usepackage{xcolor}

\usepackage{enumitem}
\usepackage{hyperref}
\hypersetup{hidelinks}
\usepackage{url}
\usepackage[nameinlink,noabbrev]{cleveref}

\newtheorem{theorem}{Theorem}
\newtheorem{proposition}{Proposition}

\newtheorem{corollary}{Corollary}
\theoremstyle{definition}

\newtheorem{assumption}{Assumption}
\theoremstyle{remark}
\newtheorem{remark}{Remark}

\DeclareMathOperator{\diag}{diag}
\DeclareMathOperator*{\argmin}{arg\,min}
\newcommand{\E}{\mathbb{E}}
\newcommand{\Prob}{\mathbb{P}}
\newcommand{\R}{\mathbb{R}}
\newcommand{\muT}{\hat{\mu}_T}
\newcommand{\muC}{\hat{\mu}_C}
\newcommand{\SigT}{\hat{\Sigma}_T}
\newcommand{\taucrm}{\hat{\tau}_{\mathrm{CRM}}}
\newcommand{\tauatt}{\tau_{\mathrm{ATT}}}
\newcommand{\tauS}{\tau_S}

\def\month{08}
\def\year{2026}
\def\openreview{\url{https://openreview.net/forum?id=z74epfCe3A}}

\title{\textbf{Centroid-Referenced Mahalanobis Matching (CRM):}\\
A Scalable, Representation-Based Framework for\\
Causal Inference in Large Observational Studies}

\author{
\name Keming Hu \email keming.hu@target.com \\
\addr Target Corporation \\
\AND
\name Yingpei He \email yingpei.he@target.com \\
\addr Target Corporation \\
}

\begin{document}
\maketitle

\begin{abstract}
Matching for causal inference can be computationally expensive at scale and
can silently change the target population when overlap is limited. We propose
\textbf{Centroid-Referenced Mahalanobis Matching (CRM)}, which replaces global
pairwise search with stratified sampling in two reference coordinates: each
unit's Mahalanobis distance from the treated centroid and its Fisher coordinate
along the treated--control mean shift. All covariates enter through the treated
covariance geometry; CRM is therefore not principal-component preprocessing
followed by nearest-neighbor matching. For $n$ units and $p$ pretreatment
covariates, its implemented cost is $O(np^2+p^3+n\log n)$, simplifying to
$O(np^2+n\log n)$ when $n\geq p$.

We derive an error decomposition separating representation, support,
discretization, and stochastic components. A pre-matching shortage fraction
$\hat\pi$ estimates the population support restriction $\pi$, which enters a
gap bound under bounded treatment-effect heterogeneity. Final retention is
reported separately for capacity-driven exclusions. Under representation
sufficiency, smoothness, and adequate cell capacity,
CRM has a conservative two-dimensional histogram mean-squared-error (MSE) bound
$O(n_T^{-1/2})$; representation sufficiency is an additional assumption, not
a consequence of ignorability given the original covariates.

On Criteo, CRM retains at least $99.4\%$ of treated units, has lower MaxSMD than
corrected propensity-score matching in $31$ of $36$ large-scale configurations,
and is roughly an order of magnitude faster. Moderate-size simulations favor
some pairwise and weighting baselines on balance, locating CRM's contribution in
scalability and explicit support diagnostics rather than universal finite-sample
dominance.
\end{abstract}


\section{Introduction}\label{sec:intro}

Observational analyses increasingly confront a practical barrier: the
datasets are too large for the methods, or the methods obscure the
inferential scope of the results.
Brute-force nearest-neighbor matching over $n_T$ treated and $n_C$ control
units requires $O(n_T n_C)$ pairwise distance evaluations
\citep{abadie2006,stuart2010}.
At modern administrative or platform-data scales, $n_T$ may be hundreds of
thousands and $n_C$ ten times larger. This cost is often prohibitive,
forcing practitioners either to subsample, approximate the search, or
abandon matching entirely in favor of regression-based alternatives.
Scalar-score matchers can reduce this cost by sorting, so computational
comparisons depend on the matcher and implementation used.

A second challenge is more subtle but equally consequential.
When treated and control covariate distributions do not overlap, every
matching procedure implicitly restricts inference to a matched subset of
treated units.
This restriction changes the target from the average treatment effect on
the treated ($\tauatt$) to a conditional average treatment effect
$\tauS$ defined on the treated population supported by the matching design
\citep{crump2009}.
Yet standard practice rarely makes this shift explicit: the caliper in
propensity score matching \citep[PSM;][]{rosenbaum1985}, the tolerance in
coarsened exact matching \citep[CEM;][]{iacus2012}, and the iteration count
in FLAME \citep{wang2021} all
silently determine who is excluded and therefore what quantity is being estimated.

This paper introduces \textbf{Centroid-Referenced Mahalanobis Matching
(CRM)}, a design-based framework that turns matching from a pairwise
unit-search problem into distributional alignment around the treated
covariate distribution. For each unit, CRM computes
$Z(X)=(d(X),\phi(X))$, where $d$ is Mahalanobis distance from the treated
centroid and $\phi$ is the Fisher direction separating treated and control
centroids in whitened covariate space. Controls are then sampled within
cells of the $(d,\phi)$ grid to reproduce the treated distribution. The
construction uses all covariates through the covariance geometry; the
two-dimensional representation is a reference coordinate system for
matching, not a low-rank replacement for the covariates followed by
nearest-neighbor matching.

The main theoretical results are an estimation-error decomposition
(\Cref{thm:bias}), an ATT gap bound under limited overlap
(\Cref{cor:gap-bound}), an $O(n_T^{-1/2})$ mean-squared-error bound in fixed
representation dimension (Proposition~\ref{prop:rate}), and
$O(np^2+p^3+n\log n)$ implemented complexity, simplifying to
$O(np^2+n\log n)$ when $n\geq p$ (Proposition~\ref{prop:complexity}). Empirically, CRM is not a
universal balance-dominance method: at moderate sample sizes, corrected PSM
or weighting methods can achieve tighter marginal balance. CRM is designed
for the complementary regime where pairwise search is costly and explicit
overlap accounting is important.

\medskip
The remainder of the paper is organized as follows.
\Cref{sec:litreview} positions CRM relative to matching, balancing scores,
and dimension-reduction baselines. \Cref{sec:setup,sec:method} define the
estimand and algorithm. \Cref{sec:theory} gives the main theoretical
properties, and \Cref{sec:simulation,sec:criteo,sec:nhanes} report the
empirical results.

\section{Related Work}\label{sec:litreview}

\subsection{Matching, Balance, and Scalability}

Matching is a design-stage strategy for making treated and control groups
more comparable before outcome modeling \citep{cochran1973,rubin1973,rubin1979,ho2007,stuart2010}.
Propensity-score methods use the balancing-score result of
\citet{rosenbaum1983} to reduce adjustment to the scalar score
$e(X)=\Prob(T=1\mid X)$, while nearest-neighbor matching has well-studied
large-sample behavior \citep{abadie2006,abadie2016}.
Other methods enforce balance more directly: CEM coarsens covariates and
matches within cells \citep{iacus2011,iacus2012}, entropy balancing matches
moments by reweighting controls \citep{hainmueller2012}, and cardinality,
optimal, or full matching solve constrained assignment problems
\citep{rosenbaum1989,hansen2004,hansen2006,zubizarreta2014}.
These methods are powerful in moderate samples, but pairwise or
combinatorial search becomes expensive in administrative and platform-scale
data.

\subsection{Learning the Matching Geometry}

Genetic Matching, MALTS, variable-importance matching, and FLAME learn or
select covariate structure to improve balance
\citep{diamond2013,wang2021,parikh2022,lanners2023}.
CRM takes a different route: the matching geometry is fixed by the treated
centroid, the treated covariance, and the Fisher direction separating
treated from control centroids. This makes the method less flexible than
learned metrics under nonlinear treatment assignment, but it avoids
outcome-dependent metric learning and eliminates global pairwise search.
The Criteo experiment illustrates why this distinction matters: methods can
achieve very low marginal imbalance while still producing large estimation
error when the matched representation does not preserve outcome-relevant
structure.

\subsection{Projection and Dimension-Reduction Baselines}\label{sec:dr-relation}

CRM is closest in appearance to ``compress-then-match'' approaches such as
principal component analysis (PCA) matching or sufficient dimension
reduction (SDR)
\citep{jolliffe2002,li1991,luo2020,brown2021pcamatchr},
but the algorithmic object is different. PCA and SDR first construct a
lower-dimensional covariate representation and then typically run an
ordinary pairwise matcher in that representation. CRM instead constructs
reference-based coordinates for every unit relative to the treated
distribution: a radial Mahalanobis coordinate $d$ and a supervised Fisher
coordinate $\phi$. Matching then occurs by stratified sampling within
$(d,\phi)$ cells. Thus the contrast with PCA is not only supervised versus
unsupervised projection; it is reference geometry plus cell-level
distributional matching versus low-dimensional preprocessing plus
nearest-neighbor search. \Cref{sec:pca-empirical} reports the corresponding
empirical comparison.

The choice of $(d,\phi)$ is deliberate. The radial coordinate is not meant
to declare two units with the same distance to be nearest neighbors in the
original Euclidean space; it stratifies units by how far they lie from the
treated reference distribution under the treated covariance geometry. The
Fisher coordinate then adds the signed direction of the treated--control
mean shift. For a binary treatment, classical linear discriminant analysis
(LDA) has at most one nonzero discriminant direction, so a ``second LDA
coordinate'' would have to come from another construction, such as PCA, SDR,
clustering, or a multi-centroid extension. Those extensions are possible,
but they move CRM toward a higher-dimensional grid and reintroduce tuning
over the number of directions.
The base method uses one radial and one mean-shift coordinate to retain a
two-dimensional matching grid.

\section{Setup and Notation}\label{sec:setup}

Let $\{(X_i, T_i, Y_i)\}_{i=1}^n$ denote $n$ independent observations,
where $X_i \in \R^p$ is a vector of pretreatment covariates, $T_i\in\{0,1\}$
is treatment assignment, and $Y_i \in \R$ is the observed outcome.
Let $n_T = \sum_i T_i$ and $n_C = n - n_T$.

\begin{assumption}[Stable Unit Treatment Value Assumption (SUTVA)]\label{ass:sutva}
$Y_i = T_i Y_i(1) + (1-T_i) Y_i(0)$, where $Y_i(0)$ and $Y_i(1)$
are the potential outcomes under control and treatment, respectively.
\end{assumption}

\begin{assumption}[Ignorability]\label{ass:ignor}
$(Y_i(0), Y_i(1)) \perp T_i \mid X_i$.
\end{assumption}

\begin{assumption}[Full-covariate overlap]\label{ass:xoverlap}
There exists $\eta>0$ such that
$\eta < \Prob(T_i=1\mid X_i=x) < 1-\eta$ for almost every $x$ in the
covariate support.
\end{assumption}

Under the potential-outcome notation, this paper focuses on the average
treatment effect on the treated,
\begin{equation}
\tauatt = \E[Y(1)-Y(0)\mid T=1].
\label{eq:att}
\end{equation}
The ATT is the natural target for one-to-one and stratified matching because
the treated sample defines the target population and controls supply
counterfactual outcomes. Extensions to average treatment effects (ATEs) are
possible by changing the reference distribution, but the present paper keeps
the target fixed to avoid mixing estimands.
The treated centroid is $\muT = n_T^{-1}\sum_{i:T_i=1} X_i$.
The treated covariance matrix with small-ridge regularization is
\begin{equation}
\SigT =
\frac{1}{n_T-1}\sum_{i:T_i=1}(X_i-\muT)(X_i-\muT)^\top + \varepsilon I_p,
\qquad \varepsilon = 10^{-8}.
\label{eq:SigT}
\end{equation}
The control centroid is $\muC = n_C^{-1}\sum_{j:T_j=0} X_j$.

\section{Method}\label{sec:method}

\subsection{CRM Representation}\label{sec:repr}

CRM follows the design-based philosophy of causal inference
\citep{rubin2007,rosenbaum2010} by separating
the \emph{design stage} (representation construction, binning, and matching)
from the \emph{estimation stage} (outcome comparison within matched cells),
so that the matched sample is constructed without access to outcomes.
CRM characterizes each unit by two scalars derived from the geometry of
the treated population.

\paragraph{Radial component.}
The Mahalanobis distance from the treated centroid is
\begin{equation}
d(x) = \bigl[(x-\muT)^\top \SigT^{-1}(x-\muT)\bigr]^{1/2}.
\label{eq:d}
\end{equation}
Under $X\mid T=1 \sim \mathcal{N}(\mu_T,\Sigma_T)$, the squared distance
$d^2(X)$ follows a $\chi^2_p$ distribution (\Cref{prop:chi}), providing
a distributional reference for the binning design.

\paragraph{Directional component.}
Let $\hat{L}$ be the lower-triangular Cholesky factor of $\SigT$, so that
$\SigT = \hat{L}\hat{L}^\top$.
Define the whitened centroid shift
\begin{equation}
w = \hat{L}^{-1}(\muC-\muT) \in \R^p.
\label{eq:w}
\end{equation}
The Fisher direction in whitened space is defined as
\begin{equation}
v =
\begin{cases}
w / \|w\|_2, & \text{if } \|w\|_2 > 0, \\
0,           & \text{if } \|w\|_2 = 0,
\end{cases}
\label{eq:v}
\end{equation}
a unit vector pointing from the treated toward the control centroid in
whitened coordinates when the centroids differ, and zero otherwise.
The directional projection of unit $x$ is then
\begin{equation}
\phi(x) = v^\top \hat{L}^{-1}(x-\muT).
\label{eq:phi}
\end{equation}
Note that $\phi(x)$ measures how far unit $x$ lies along the treated--control
mean-shift direction in whitened space.
It is well-defined whenever $\muC \neq \muT$; when the centroids coincide,
$w = 0$ and the directional correction is unnecessary (\Cref{prop:indep}).

\paragraph{Why this representation.}
The two coordinates target complementary geometric discrepancies. The
Mahalanobis radius $d$ uses every covariate through the treated covariance
metric and places observations on covariance-aligned shells around the treated
centroid, as in classical Mahalanobis matching \citep{rubin1979}. The Fisher
coordinate $\phi$ then distinguishes locations along the mean-separating axis
that the radius alone cannot distinguish. Under the homoskedastic
Gaussian-class model, this axis is proportional to the population
propensity-score gradient (\Cref{eq:phips}). The Fisher-coordinate CRM variant
is therefore a deliberately simple default for settings dominated by a linear
mean shift, while retaining the computational benefit of fixed-dimensional
binning. This construction is not PCA and does not claim that $(d,\phi)$
preserves all information in $X$ or is a balancing score: many covariate
vectors share the same representation, which is why identification requires
Assumption~\ref{ass:reprignor}. Nonlinear assignment or treatment-dependent
covariance can require additional coordinates or a different reference
representation (\Cref{sec:discussion}).

\paragraph{CRM summary.}
The two-dimensional representation is $Z(x) = (d(x),\phi(x)) \in \R^2$.
We use $k$ to denote the number of directional coordinates included:
$k=0$ is the radial-only representation $Z(x)=d(x)$, and $k=1$ is the
default CRM representation using both $d$ and the Fisher coordinate $\phi$.
Both components are computed independently per unit in $O(p^2)$ after
the one-time $O(n_T p^2)$ estimation of $\SigT$ and $O(p^3)$ Cholesky
factorization.
Throughout the paper, $n=n_T+n_C$ denotes total sample size and $p$ denotes
the number of pretreatment covariates.
\Cref{fig:intuition} illustrates the representation geometrically.

\begin{figure}[htbp]
\centering
\includegraphics[width=\linewidth]{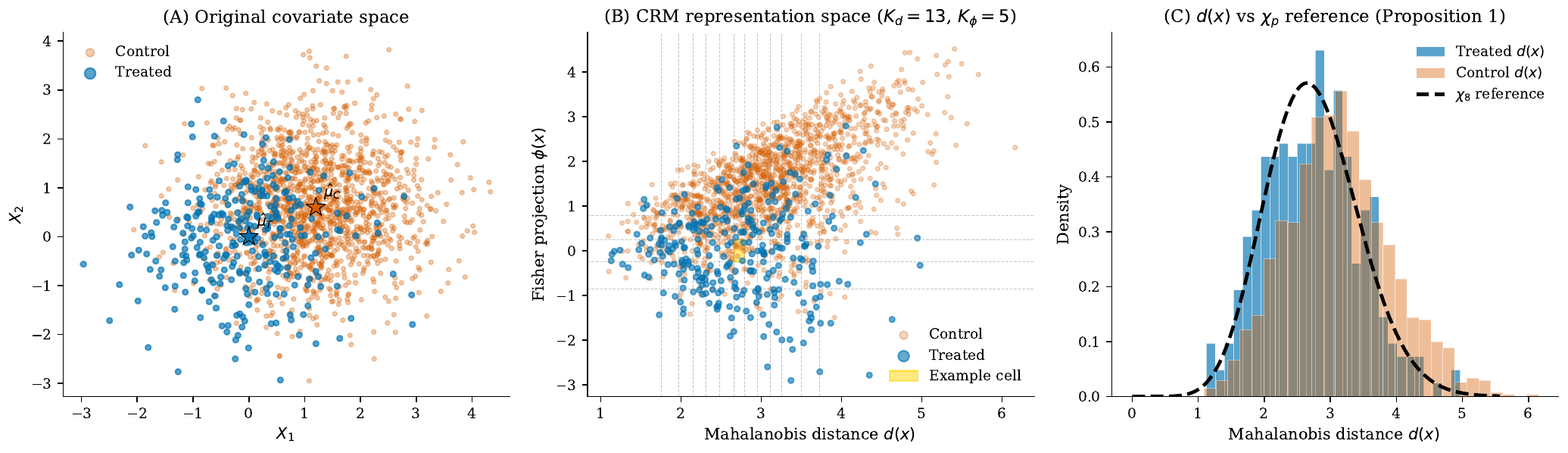}
\caption{CRM geometric intuition ($n_T=300$, $n_C=1{,}500$, $p=8$,
centroid shift $\approx1.3$).
\textbf{(A)}~Original covariate space: treated (blue) and control (orange)
clouds with sample centroids~$\hat\mu_T$ and~$\hat\mu_C$ marked.
\textbf{(B)}~CRM representation space $(d,\phi)$ with equal-frequency grid
overlaid; the gold cell illustrates one matched stratum.
\textbf{(C)}~Treated distance distribution versus the $\chi_p$ reference
(Proposition~\ref{prop:chi}), confirming the distributional foundation
for equal-frequency binning.}
\label{fig:intuition}
\end{figure}

\paragraph{Role of the Fisher coordinate.}
The radial-only variant ($k=0$, using $d$ alone without $\phi$) is
motivated by the special case in which no treated--control mean-shift
direction needs to be represented.
When $\mu_C \neq \mu_T$, radial binning cannot distinguish controls on
opposite sides of the same covariance-aligned shell, so systematic
directional imbalance can remain.
The default throughout this paper is therefore the Fisher-coordinate CRM
variant; \Cref{sec:sim-results} evaluates the difference empirically.

\subsection{Pre-Matching Support Diagnostic}\label{sec:diag}

Before any matching is performed, CRM constructs a grid on $(d,\phi)$
and counts control units in each cell.
Bin edges for the distance axis are set using the Freedman--Diaconis
rule \citep{freedman1981} applied to the treated distance distribution:
\begin{equation}
K_d = \min\!\left\{200,\,\max\!\left(10,
\left\lceil\frac{\max(d_T)-\min(d_T)}{2\,\mathrm{IQR}(d_T)\,n_T^{-1/3}}
\right\rceil\right)\right\},
\label{eq:Kd}
\end{equation}
where $d_T = \{d(X_i): T_i=1\}$.
The number of directional bins is $K_\phi = \max\{5, \lfloor K_d/4\rfloor\}$.
Bin edges for each axis are set at equal-frequency quantiles of the
\emph{marginal} treated distribution on that axis: $K_d$ quantiles of
$\{d(X_i): T_i=1\}$ and $K_\phi$ quantiles of
$\{\phi(X_i): T_i=1\}$.
This ensures each marginal bin contains approximately $n_T/K_d$ or
$n_T/K_\phi$ treated units respectively.
It does not equalize joint cell counts: $d$ and $\phi$ are generally
dependent, including under the Gaussian reference model. Joint cells can
therefore remain sparse, which is precisely what $\hat\pi$ is designed to
detect.

The \emph{shortage fraction} is
\begin{equation}
\hat\pi =
\frac{1}{n_T}\sum_{i:T_i=1}
\mathbf{1}\!\bigl(|C(k_d(X_i),k_\phi(X_i))|=0\bigr),
\label{eq:pihat}
\end{equation}
where $C(k_d,k_\phi)$ denotes the set of control units assigned to cell
$(k_d,k_\phi)$, and $k_d(X_i)$, $k_\phi(X_i)$ are the cell indices of
unit~$i$.
A cell is \emph{unsupported} if it contains at least one treated unit but
no controls; $\hat\pi$ is the fraction of treated units in unsupported cells.
This diagnostic is computed and reported before any units are discarded.

Calibration simulations ($n_T=1{,}000$, $p=10$, $n_C=10{,}000$) show that
$\hat\pi$ increases sigmoidally as a function of the centroid shift magnitude
$\|\mu_C-\mu_T\|$, crossing the $5\%$ threshold at shift $\approx 0.54$
and the $30\%$ threshold at shift $\approx 0.71$.
The Rayleigh statistic \citep{mardia2000} (\Cref{prop:indep}) remains flat
across all shift magnitudes, confirming it measures directional non-uniformity
within the treated group rather than inter-group separation; the two
diagnostics thus provide complementary information (\Cref{fig:shortage}).

\begin{figure}[htbp]
\centering
\includegraphics[width=\linewidth]{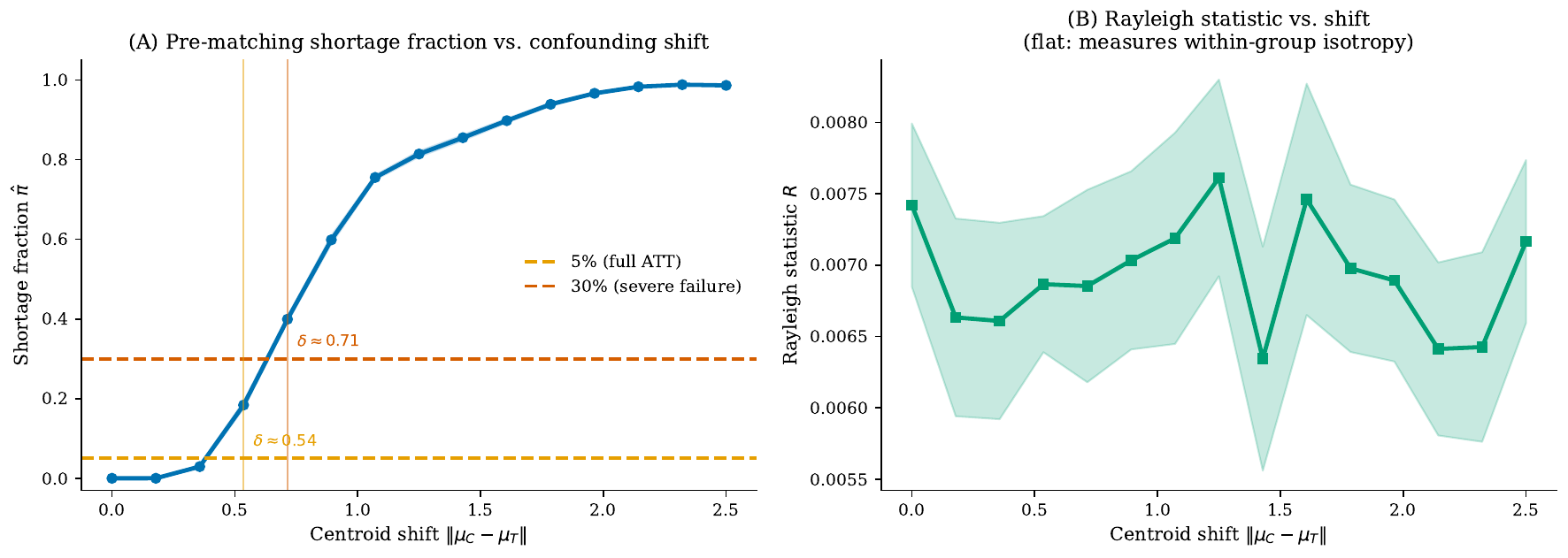}
\caption{Pre-matching shortage diagnostic ($n_T=1{,}000$, $p=10$,
$n_C=10{,}000$; shaded band = mean $\pm$ 95\% CI over replications).
\textbf{(A)}~Shortage fraction $\hat\pi$ increases sigmoidally with the
centroid shift magnitude, crossing 5\% at $\delta\approx0.54$ and 30\% at
$\delta\approx0.71$.
\textbf{(B)}~Rayleigh statistic $R$ remains flat, confirming it measures
within-group isotropy of the treated distribution rather than the
inter-group separation that drives $\hat\pi$; the two diagnostics are
complementary.}
\label{fig:shortage}
\end{figure}

\subsection{Stratified Matching Estimator}\label{sec:estimator}

For each cell $(k_d, k_\phi)$, let $T(k)$ and $C(k)$ denote the sets of
treated and control units assigned to that cell.
CRM samples $n_k = \min\{|T(k)|, |C(k)|\}$ controls without replacement
from $C(k)$.
Define the \emph{supported treated set}
\begin{equation}
S = \bigl\{i : T_i=1,\; |C(k_d(X_i),k_\phi(X_i))|>0\bigr\}.
\label{eq:S}
\end{equation}
For finite-sample matching without replacement, let $S_M\subseteq S$ denote
the treated units actually retained after cell-capacity constraints; when
every supported treated unit is retained, $S_M=S$.
The CRM estimator is
\begin{equation}
\taucrm =
\frac{1}{|S_M|}\sum_{i\in S_M}
\Biggl[Y_i -
\frac{1}{|C_i|}\sum_{j\in C_i} Y_j\Biggr],
\label{eq:taucrm}
\end{equation}
where $C_i$ denotes the set of matched controls in the same cell as treated
unit~$i$.
Each treated unit is compared to the cell-mean outcome of its matched
controls.
The shortage fraction $\hat\pi$ measures zero-control home cells before
matching; it does not include additional attrition when a supported cell has
fewer controls than treated units. We therefore report both $\hat\pi$ and
retention $|S_M|/n_T$.
Uncertainty is quantified by a paired bootstrap with $B=500$ replications
\citep{efron1993}, which we report as an \emph{approximate} measure of
variability.
A calibration study finds sub-nominal coverage for the paired bootstrap
across the tested sample sizes (\Cref{fig:coverage}); the conservative
cell-stratified and bias-correction diagnostics in \Cref{app:variance}
show that variance calibration and representation bias can both affect
coverage.

\begin{center}
\noindent\begin{minipage}{0.97\linewidth}
\small\setlength{\fboxsep}{8pt}
\fbox{\begin{minipage}{0.95\linewidth}
\textbf{Algorithm~1: CRM Matching ($k=1$)}\\[0.4em]
\textbf{Input:} Covariates $X_T, X_C$; outcomes $Y_T, Y_C$ withheld until Stage 3.\\
\textbf{Output:} Matched sample; $\hat\pi$; retention; $\taucrm$ with approximate bootstrap uncertainty summary (see \Cref{app:variance}).\\[0.5em]
\textbf{Stage 1: Representation and diagnostic} [$O(np^2+p^3+n\log n)$]\\
1.\ Compute $\muT, \SigT$ via \Cref{eq:SigT}; Cholesky-factor $\SigT = \hat{L}\hat{L}^\top$.\\
2.\ Compute $\muC$; set $w = \hat{L}^{-1}(\muC-\muT)$; if $\|w\|>0$ set $v = w/\|w\|_2$.\\
3.\ Compute $d(x_i)$ via \Cref{eq:d} and $\phi(x_i)$ via \Cref{eq:phi} for all $n$ units.\\
4.\ Set $K_d$ via \Cref{eq:Kd}; set $K_\phi = \max\{5, \lfloor K_d/4\rfloor\}$.\\
5.\ Assign units to cells; count controls per cell; report $\hat\pi$ via \Cref{eq:pihat}.\\[0.5em]
\textbf{Stage 2: Matching} [$O(n)$]\\
6.\ For each cell with $|T(k)|>0$: if $|C(k)|=0$, flag treated units as
    unsupported; otherwise form up to $n_k$ pairs without replacement
    (random or NN; see \Cref{sec:wcnn}) and record unmatched treated units.\\
7.\ Return matched indices; report $\hat\pi$, retention, and covariate profiles
    of unsupported and capacity-excluded treated units.\\[0.5em]
\textbf{Stage 3: Estimation} [$O(|S_M|)$]\\
8.\ Compute $\taucrm$ via \Cref{eq:taucrm}; construct a paired-bootstrap uncertainty summary ($B=500$; approximate; see \Cref{app:variance}).
\end{minipage}}
\end{minipage}
\end{center}

\subsection{Within-Cell Nearest-Neighbor Refinement (CRM-NN)}\label{sec:wcnn}

Standard CRM selects each matched control by uniform random sampling within
its cell (\textsc{CRM-random}).
Under the shrinking-cell conditions used for consistency, differences among
controls in the same cell vanish asymptotically. At small $n_T$, however,
cells can be wide and random sampling discards within-cell precision.

A refinement replaces uniform sampling with nearest-neighbor selection in
the whitened full-$p$-dimensional space, restricted to the cell's candidate
pool.
We call this \textsc{CRM-NN}.
Concretely, for treated unit $i$ in cell $k$, we select
\begin{equation}
j^*(i) = \argmin_{j \in C(k),\, j \notin \text{used}}
\bigl\|\hat{L}^{-1}(X_j - X_i)\bigr\|_2.
\label{eq:wcnn}
\end{equation}
The within-cell cost is $O(|C(k)| \cdot p)$ per cell; summing over all
cells gives $O(n_C p)$ additional work.
Thus the total implemented complexity of CRM-NN is
$O(np^2+p^3+n\log n+n_Cp)$; when $n_C/n_T$ is bounded and $n\geq p$, this has the same
leading linear-algebra order as base CRM, but with a larger constant.
CRM-NN overhead matters in practice when cells are large or the
control-to-treated ratio is very high.

A timing experiment across $n_T \in \{185, \ldots, 50{,}000\}$, summarized in
\Cref{fig:crm-nn}, shows that
CRM-NN yields $8$--$30\%$ MaxSMD improvement over CRM-random across all
tested sample sizes under a $10\!:\!1$ control ratio, with runtime overhead
small enough for interactive use when $n_T \leq 1{,}000$.
On the LaLonde CPS benchmark \citep{lalonde1986} ($n_T=185$, $n_C=15{,}992$,
true ATT $=\$1{,}794$ from the NSW randomized experiment
\citep{dehejia1999}), CRM-NN reduces ATT bias from
$-\$908$ to $+\$135$ while retaining $84\%$ of treated units, superior to
full Mahalanobis NN ($-\$199$ bias, $67\%$ retention) on both criteria
simultaneously (\Cref{fig:lalonde}).

\begin{figure}[htbp]
\centering
\includegraphics[width=\linewidth]{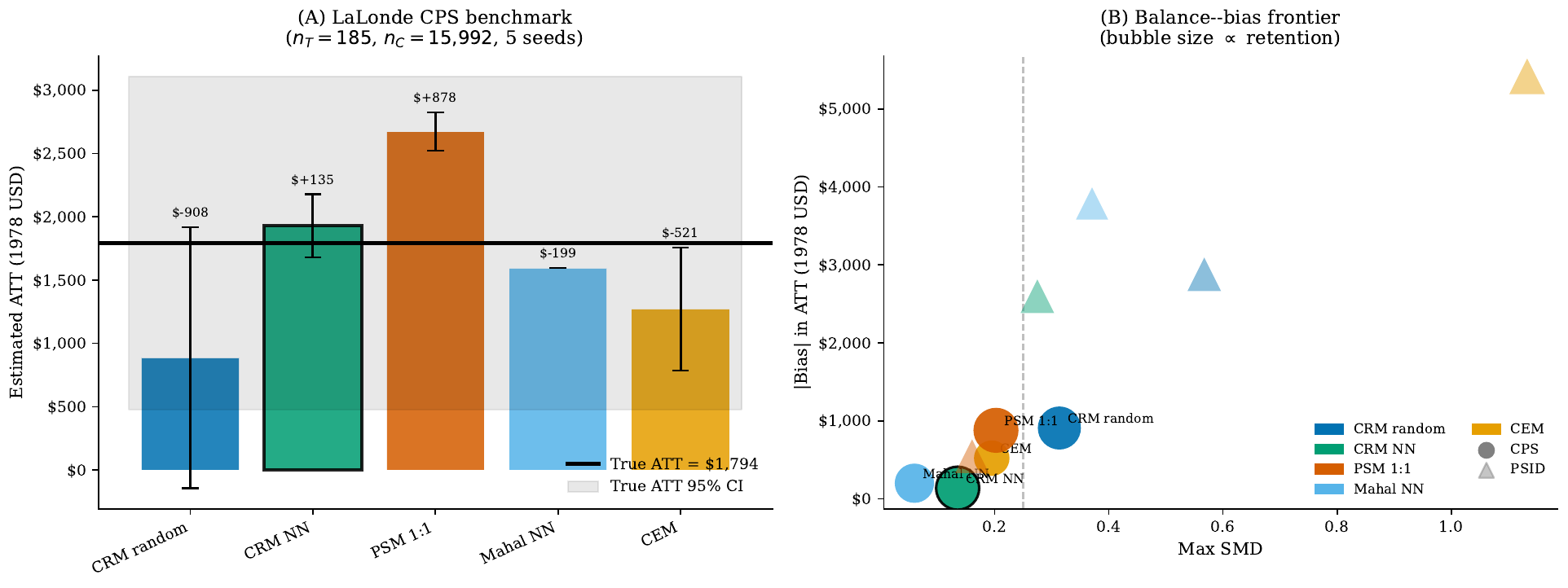}
\caption{LaLonde CPS benchmark \citep{lalonde1986} (true ATT $=\$1{,}794$
from the NSW randomized experiment; $n_T=185$, $n_C=15{,}992$; 5 seeds).
\textbf{(A)}~ATT point estimates with 95\% CI; bias (relative to true ATT)
annotated above each bar; black outline marks CRM-NN.
\textbf{(B)}~Balance--bias frontier; bubble size proportional to treated
retention.
CRM-NN occupies a lower-bias position than CRM-random while retaining
the same $84\%$ of treated units, and achieves lower bias with higher
retention than full Mahalanobis NN.}
\label{fig:lalonde}
\end{figure}

\begin{remark}
\textbf{Two-regime recommendation.}
Use CRM-NN when $n_T < 1{,}000$ (cells are large, local precision matters,
NN overhead $<100$~ms is negligible).
Use CRM-random when $n_T \geq 1{,}000$ (cells shrink, random $\approx$ NN
in quality, and NN overhead grows proportionally to $n_C/n_T$).
\end{remark}

\subsection{Connection to Propensity Score Methods}\label{sec:psconnect}

Under a linear treatment assignment model
$\Prob(T=1\mid X) = \sigma(\beta^\top X)$ with equal treated and control
covariance matrices, the LDA weight vector satisfies
$\beta \propto \Sigma^{-1}(\mu_C-\mu_T)$ \citep{fisher1936}.
Substituting into \Cref{eq:phi}:
\begin{align}
\phi(x)
&= \frac{(\muC-\muT)^\top \SigT^{-1}(x-\muT)}
        {\|\hat{L}^{-1}(\muC-\muT)\|_2}
\;\propto\; \beta^\top x + c
\label{eq:phips}
\end{align}
for some constant $c$ determined by the centroid shift.
CRM distributional matching on $(d,\phi)$ therefore simultaneously corrects
the propensity-score dimension (via $\phi$) and the Mahalanobis spread
(via $d$), without fitting a treatment model.
This is a geometric analogy, not a general identity: the proportionality
$\phi(x)\propto\beta^\top x + c$ holds under the homoskedastic Gaussian-class
model underlying linear discriminant analysis.
Outside that setting (for instance, under logistic treatment assignment with
non-Gaussian covariates), the logistic regression coefficient vector need not
equal the LDA direction.
$\phi(x)$ should therefore be interpreted as a directional score that
approximates the propensity-score dimension, not as a provably exact
substitute.
CRM uses $v$ regardless of the true treatment mechanism; its validity rests
on Assumption~\ref{ass:reprignor}, not on any treatment-model specification.

\subsection{CRM Variants}\label{sec:variants}

Three practical variants address specific data configurations; full
descriptions are given in \Cref{app:variants}.
\textbf{CRM-CAM} scales covariates to unit variance before computing
Mahalanobis distances, recommended for mixed continuous/binary data.
\textbf{CRM-Pool} uses pooled bin edges, reducing attrition under large
centroid shifts.
\textbf{CRM-Trim} explicitly discards treated units outside the control
distance support and reports them separately with an estimand caveat.

\section{Theoretical Properties}\label{sec:theory}

\subsection{Target Estimand and Supported Population}\label{sec:estimand}

CRM operates under potentially limited overlap by explicitly restricting
inference to the subset of treated units for which valid counterfactuals
exist in the representation space.
Recall from \Cref{eq:S} that the supported treated set is
\[
S = \bigl\{ i : T_i=1,\;\; |C(k_d(X_i), k_\phi(X_i))| > 0 \bigr\}.
\]
CRM therefore targets the \emph{conditional average treatment effect on the
supported treated} (CATT):
\begin{equation}
\tauS = \E[Y(1) - Y(0) \mid T = 1,\; i \in S],
\label{eq:tauS-def}
\end{equation}
which coincides with the full $\tauatt$ only when overlap is complete
($\hat\pi = 0$; \citealp{crump2009}).

\medskip\noindent
This formulation makes explicit a feature that is implicit in \emph{all}
matching methods: when overlap fails, every matching procedure restricts
the estimand to a subset of treated units.
CRM makes this restriction visible before any units are discarded, via the
pre-matching shortage fraction $\hat\pi$ (\Cref{eq:pihat}).

\subsection{Identification Under Representation}\label{sec:ident}

CRM replaces conditioning on the full covariate vector $X$ with conditioning
on the low-dimensional representation $Z(X) = (d(X), \phi(X))$.

Assumption~\ref{ass:ignor} (ignorability given $X$), together with
full-covariate overlap, justifies the causal problem at the full-covariate
level but is not sufficient for CRM's identification: a further
representation-level assumption is required. Assumption~\ref{ass:reprignor}
below strengthens this to ignorability given $Z(X)$; it holds exactly when
$(d,\phi)$ is sufficient for treatment assignment and the potential outcomes.
The Fisher direction targets the linear treated--control mean difference but
does not establish this sufficiency.

\begin{assumption}[Representation Sufficiency]\label{ass:reprignor}
$(Y(0), Y(1)) \perp T \mid Z(X)$, where $Z(X) = (d(X), \phi(X))$.
\end{assumption}

\begin{assumption}[Overlap in Representation Space]\label{ass:overlap}
$0 < \Prob(T=1 \mid Z(X)) < 1$ for all $Z(X)$ in the support
$\mathcal{Z}_S$ of the supported treated distribution.
\end{assumption}

\begin{assumption}[Smoothness]\label{ass:smooth}
There exists $L<\infty$ such that, for all $z,z'\in\mathcal{Z}_S$,
\[
\left|\E[Y(0)\mid Z=z]-\E[Y(0)\mid Z=z']\right|
\leq L\|z-z'\|_2 .
\]
\end{assumption}

\begin{proposition}[Identification]\label{prop:ident}
Under Assumptions~\ref{ass:sutva}, \ref{ass:reprignor},
and~\ref{ass:overlap},
\[
\tauS =
\E\!\Bigl[
\E[Y \mid T=1, Z] - \E[Y \mid T=0, Z]
\;\Big|\; T=1,\; Z \in \mathcal{Z}_S
\Bigr].
\]
\end{proposition}

\begin{proof}[Sketch]
By consistency, $\E[Y\mid T=t,Z]=\E[Y(t)\mid T=t,Z]$.
Assumption~\ref{ass:reprignor} then gives
$\E[Y(t)\mid T=t,Z]=\E[Y(t)\mid Z]$ for $t\in\{0,1\}$.
Taking their difference and averaging over the supported treated
distribution identifies $\tau_S$. This proposition concerns the population
estimand; consistency of the cell estimator is addressed separately below.
\end{proof}

\begin{remark}
Identification shifts from high-dimensional covariate adjustment to
\emph{representation adequacy}: CRM is identified when $(d,\phi)$ captures
all confounding variation relevant for outcome differences.
The Fisher direction $v = w/\|w\|_2$ maximizes the separation between
treated and control group means in whitened covariate space, capturing
the dominant \emph{linear} component of treatment assignment heterogeneity.
Formally, $v$ solves
$\max_{\|u\|=1} (u^\top w)^2$, so it points exactly along the direction
of the centroid shift $\muC - \muT$ in the whitened space.
Thus $\phi$ captures the \emph{linear} confounding component associated
with the treatment-group mean difference; it does not, by itself,
guarantee that conditioning on $(d,\phi)$ removes all confounding, which is
the content of Assumption~\ref{ass:reprignor} and is not implied by
ignorability given $X$ (Assumption~\ref{ass:ignor}) alone.
In particular, $v$ may fail to represent nonlinear or higher-order
confounding that is orthogonal to the mean shift; such residual confounding
contributes to the representation bias in \Cref{thm:bias}.
\end{remark}

\begin{proposition}[Consistency]\label{prop:consistency}
Under Assumptions~\ref{ass:sutva}, \ref{ass:reprignor}, \ref{ass:overlap},
and~\ref{ass:smooth}, with $n_C/n_T \to r > 0$ and an uncapped sequence of
grids satisfying $K_dK_\phi\to\infty$, maximal cell diameter $h_n\to0$,
minimum supported-cell counts tending to infinity, and capacity-induced
attrition $|S\setminus S_M|/n_T\to0$,
\[
\taucrm \xrightarrow{p} \tauS \quad \text{as } n_T \to \infty.
\]
\end{proposition}

\begin{proof}[Sketch]
These conditions make the grid finer while retaining enough treated and
control observations in each supported cell. The uncapped
Freedman--Diaconis scaling has $K_d\propto n_T^{1/3}$; with
$K_\phi\propto K_d$, the number of joint cells is $O(n_T^{2/3})$ and
average cell counts are $O(n_T^{1/3})$.
The implementation caps $K_d$ at 200 as a finite-sample safeguard. That cap
must increase with sample size for this asymptotic statement; it is not part
of the grid sequence assumed by the proposition.
These are the standard ingredients used to establish consistency for
histogram, CEM \citep{iacus2012}, and kernel matching estimators
\citep{heckman1997}.
As $n_T \to\infty$, the grid diameter $h_n \to 0$ and each cell converges
to a point mass.
The cell control means converge to $\E[Y(0)\mid Z=z_k]$ by the law of
large numbers, and the bias $O(h_n) \to 0$ by Assumption~\ref{ass:smooth}.
\end{proof}

\subsection{Estimand Gap Under Limited Overlap}\label{sec:gap}

\begin{theorem}[Conditional ATT decomposition]\label{thm:estimand-gap}
Let $\pi = \Prob(i \notin S\mid T_i=1)$ denote the population shortage
fraction, $\tauS = \E[Y(1)-Y(0)\mid T=1,\; i\in S]$, and
$\tau_{S^c} = \E[Y(1)-Y(0)\mid T=1,\; i\notin S]$.
Then
\begin{equation}
\tauatt = (1-\pi)\tauS + \pi\,\tau_{S^c}.
\label{eq:att-decomp}
\end{equation}
Consequently,
\begin{equation}
\tauatt - \tauS = \pi\bigl(\tau_{S^c} - \tauS\bigr),
\label{eq:gap}
\end{equation}
which equals zero if and only if $\pi=0$ or $\tau_{S^c} = \tauS$.
\end{theorem}

\begin{proof}
Partition the treated population on $\{i\in S\}$ and apply the law of
total expectation:
\begin{align*}
\tauatt
&= \E[Y(1)-Y(0)\mid T=1] \\
&= \E[Y(1)-Y(0)\mid T=1,\; i\in S]\,(1-\pi)
 + \E[Y(1)-Y(0)\mid T=1,\; i\notin S]\,\pi \\
&= (1-\pi)\tauS + \pi\,\tau_{S^c}.
\end{align*}
Rearranging gives \Cref{eq:gap}. \qed
\end{proof}

\begin{corollary}[Gap bound]\label{cor:gap-bound}
Let $M_\tau=\sup_i\{Y_i(1)-Y_i(0)\}-
\inf_i\{Y_i(1)-Y_i(0)\}$ be the range of individual treatment effects.
Assume $M_\tau<\infty$.
Then
\begin{equation}
|\tauatt - \tauS| \leq M_\tau\pi.
\label{eq:gap-bound}
\end{equation}
\end{corollary}

\begin{proof}
Both conditional average treatment effects lie between the minimum and
maximum individual treatment effects, so
$|\tau_{S^c}-\tau_S|\leq M_\tau$. Apply \Cref{eq:gap}.
\end{proof}

\begin{remark}
\Cref{thm:estimand-gap} applies to \emph{any} matching estimator.
The value of CRM's pre-matching diagnostic is that $\hat\pi$ is reported
before any unit is discarded. This makes the support-restriction weight
observable; the gap itself still depends on treatment-effect heterogeneity.
When $M_\tau$ can be bounded from prior knowledge or sensitivity analysis,
\Cref{eq:gap-bound} gives an explicit quantitative bound on estimand drift.
\end{remark}

\subsection{Bias Decomposition}\label{sec:bias}

The main theoretical contribution of this section is a four-part
decomposition of the total estimation error of $\taucrm$ relative to
$\tauatt$.
The decomposition clarifies why low marginal imbalance need not imply low
estimation error.

\begin{theorem}[Estimation error decomposition]\label{thm:bias}
Let $\Delta(z)$ denote residual within-representation confounding:
\[
\Delta(z)
= \E\bigl[Y(0) \mid Z(X)=z,\; T=1\bigr]
- \E\bigl[Y(0) \mid Z(X)=z,\; T=0\bigr].
\]
Let $h_n$ be the maximal cell diameter.
Under Assumptions~\ref{ass:sutva}, \ref{ass:ignor}, and
\ref{ass:smooth}, and with $S_M=S$, the total estimation error of the CRM
estimator admits
the decomposition
\begin{align}
\taucrm - \tauatt
&=
\underbrace{\E\bigl[\Delta(Z(X)) \mid T=1,\; i\in S\bigr]}_{\text{representation bias (residual confounding within $Z$)}}
\label{eq:full-decomp}\\
&\quad+
\underbrace{\tauS-\tauatt}_{\text{support restriction bias } = -\pi(\tau_{S^c}-\tauS)}
\notag\\
&\quad+ \underbrace{O_p(h_n)}_{\text{cell approximation}}
+ \underbrace{O_p\!\left((n_T h_n^2)^{-1/2}\right)}_{\text{conservative stochastic bound}}.
\notag
\end{align}
The representation-bias term is zero under
Assumption~\ref{ass:reprignor}.
If $S_M\neq S$, an additional finite-sample subset-selection term
$\tau_{S_M}-\tau_S$ enters the decomposition, where
$\tau_{S_M}=\E[Y(1)-Y(0)\mid T=1,i\in S_M]$. This term must be assessed
through retention and the characteristics of excluded treated units.
\end{theorem}

\begin{proof}[Sketch]
Add and subtract $\tauS$:
\begin{align*}
\taucrm - \tauatt
&= (\taucrm - \tauS) + (\tauS - \tauatt).
\end{align*}
The second term equals $-\pi(\tau_{S^c}-\tauS)$ by
\Cref{thm:estimand-gap}; this is the support restriction bias.
For the first term, within each cell $k$ the control mean estimates
$\E[Y(0)\mid Z=z_k, T=0]$.
The corresponding treated counterfactual mean is
$\E[Y(0)\mid Z=z_k, T=1]$; their difference is $\Delta(z_k)$.
Finite cell diameter contributes $O_p(h_n)$ under
Assumption~\ref{ass:smooth}, and the $O_p(n_T^{-1/2})$ term collects
within-cell sampling variability.
\end{proof}

\begin{remark}
\Cref{thm:bias} separates four sources of error distinctly:
\emph{representation bias}, \emph{support restriction bias},
\emph{cell-approximation error} (the deterministic $O(h_n)$ discretization
term from finite cell diameter), and \emph{stochastic error} (the
$O_p((n_T h_n^2)^{-1/2})$ conservative within-cell sampling bound).
The \emph{representation bias} $\E[\Delta(Z(X))\mid T=1,i\in S]$
is zero under Assumption~\ref{ass:reprignor}; the Fisher direction $v$ is
designed to capture the treated--control mean-shift component, but it need
not remove nonlinear or higher-order confounding.
The \emph{support restriction bias} $\pi(\tau_{S^c}-\tauS)$ is zero
if and only if $\pi=0$ or $\tau_{S^c}=\tau_S$;
it is bounded by $M_\tau\pi$ (\Cref{cor:gap-bound}) and diagnosed by
the pre-matching diagnostic.
The \emph{cell-approximation error} vanishes as the grid refines, and the
conservative stochastic term vanishes when $n_T h_n^2\to\infty$.
The decomposition is specific to CRM, but its broader lesson is that low
marginal MaxSMD alone does not control representation or support-restriction
bias. The CEM result in \Cref{sec:criteo} is an empirical illustration of
that distinction, not a consequence of this theorem for CEM.
\end{remark}

\begin{proposition}[Representation error]\label{prop:bias}
Under Assumptions~\ref{ass:sutva} and~\ref{ass:ignor}, the quantity
$\Delta(z)$ in \Cref{thm:bias} is the bias from replacing full-covariate
adjustment by adjustment on $Z(X)$.
Under Assumption~\ref{ass:reprignor}, $\Delta(z)=0$ for all supported
$z$; in general, $\Delta(z)$ measures the residual confounding not captured
by the CRM representation.
\end{proposition}

\subsection{Supporting Geometric Properties}\label{sec:chi}

Two additional results support the design choices in Algorithm~1; proofs
are in \Cref{app:geometry}.

Under a Gaussian treated distribution, $d^2(X)\mid T=1 \sim \chi^2_p$
(\Cref{prop:chi}), which supplies a radial reference distribution.
The implementation uses empirical treated-sample quantiles, rather than
population $\chi^2_p$ quantiles, to obtain approximately equal marginal
treated counts per distance bin.
Separately, $d(X)$ and the unit direction $u(X) = \hat{L}^{-1}(X-\mu_T)/d(X)$
are independent under the Gaussian model (\Cref{prop:indep}). This polar
decomposition separates radial magnitude from unit direction; the actual
coordinate $\phi=d\,v^\top u$ combines both through a signed projection and
is not independent of $d$.
This within-treated independence is unaffected by $\mu_C$. The role of
$\phi$ is instead to distinguish the angular direction of the
treated--control centroid shift, information that $d$ alone discards.
We stress that these two properties are \emph{motivational}: they
characterize the reference geometry of the treated distribution around its
own centroid and justify the binning and the inclusion of $\phi$, but they
are not identification guarantees and do not imply that $(d,\phi)$ are
independent (or that confounding is removed) in the matched
treated-versus-control comparison, which is governed instead by
Assumption~\ref{ass:reprignor} and the bias decomposition of
\Cref{thm:bias}.

Algorithm~1 has implemented time complexity $O(np^2+p^3+n\log n)$ and space
complexity $O(p^2+n)$, both independent of pairwise treated--control search
(\Cref{prop:complexity}).
In the paper's $n\gg p$ regime, $O(np^2)$ dominates the one-time $O(p^3)$
factorization; $O(n\log n)$ comes from quantile binning.
Brute-force nearest-neighbor matching adds a pairwise-search term, whereas
one-dimensional propensity-score matching can also exploit sorting. The
runtime comparison below is therefore empirical and implementation-specific.
Measured $1.5$--$2.7\times$ speedups over PSM at $n_T \leq 5{,}000$
widen rapidly at larger scales (\Cref{fig:runtime}).
Because CRM's dominant operations are covariance estimation, Cholesky
factorization, and vectorized matrix--vector products (all standard dense
linear algebra), the implementation is compatible with GPU-accelerated
backends (e.g.\ CuPy, JAX).
A systematic GPU benchmark is left to future work.

\begin{figure}[htbp]
\centering
\includegraphics[width=\linewidth]{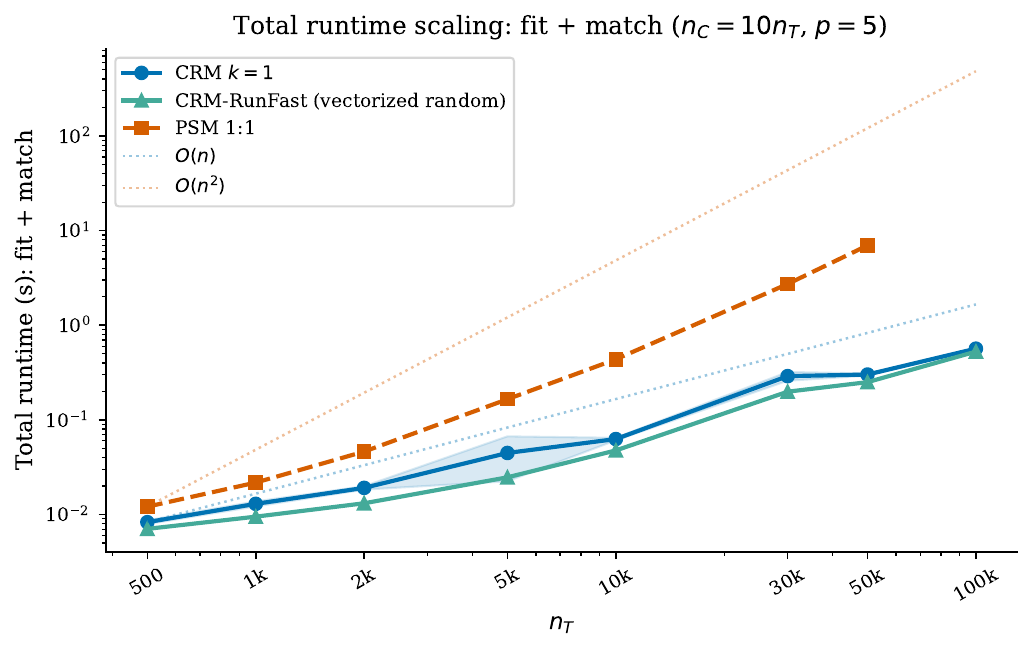}
\caption{Runtime scaling on log--log axes ($n_C=10n_T$; 3 replicates).
The extended sweep adds larger $n_T$ values. CRM follows the $O(n)$ reference
line over this range, whereas the evaluated PSM implementation curves upward.
The comparison describes these implementations; efficient scalar-score
matchers can also exploit sorting.}
\label{fig:runtime}
\end{figure}

\subsection{Conservative Rate Bound}\label{sec:rate}

\noindent\begin{minipage}{\linewidth}
\small\setlength{\fboxsep}{7pt}
\fbox{\begin{minipage}{0.97\linewidth}
\textbf{Assumptions $\Rightarrow$ implication.}\quad
The $O(n_T^{-1/2})$ MSE upper bound in \Cref{prop:rate} requires:
\textbf{SUTVA} (\Cref{ass:sutva}), 
\textbf{representation sufficiency}
$(Y(0),Y(1))\perp T\mid Z(X)$ (\Cref{ass:reprignor}),
and \textbf{Lipschitz smoothness} of $\E[Y(0)\mid Z]$ (\Cref{ass:smooth}).
Under these assumptions, matching in the fixed two-dimensional space $(d,\phi)$
yields a histogram bound governed by representation dimension $d_Z=2$
rather than the original covariate dimension $p$. This gain is conditional on
representation sufficiency; unlike propensity-score sufficiency, it does not
follow from ignorability given $X$.
The assumption most likely to fail in practice is representation sufficiency:
the bias constant
$\Delta(z)$ in \Cref{prop:bias} quantifies residual confounding when this
assumption holds only approximately.
\end{minipage}}
\end{minipage}\medskip

\begin{proposition}[Conservative fixed-dimension rate bound]\label{prop:rate}
Under Assumptions~\ref{ass:sutva}, \ref{ass:reprignor},
and~\ref{ass:smooth}, with $d_Z=2$, $n_C/n_T \to r \in (0,\infty)$,
$h_n\to0$, $n_T h_n^2\to\infty$, and no capacity-induced attrition
($S_M=S$),
the mean squared error of $\taucrm$ satisfies
\begin{equation}
\mathrm{MSE}(\taucrm)
= O\!\left(h_n^2\right) + O\!\left(\frac{1}{n_T h_n^2}\right).
\label{eq:mse}
\end{equation}
Balancing the two terms at $h_n = n_T^{-1/4}$ gives
\begin{equation}
\mathrm{MSE}(\taucrm) = O\!\left(n_T^{-1/2}\right).
\label{eq:rate}
\end{equation}
For comparison, the same conservative histogram argument in the full
$p$-dimensional covariate space gives
$\mathrm{MSE} = O(n_T^{-2/(2+p)})$, which deteriorates rapidly
with $p$.
Under the same conservative histogram bound, the rate ratio is
$n_T^{(p-2)/(2(2+p))}$, which diverges for $p>2$;
for $p=10$ the gain is $n_T^{1/3}$.
\end{proposition}

\begin{proof}
See \Cref{app:proofs}.
\end{proof}

\begin{remark}
The $O(n_T^{-1/2})$ result is the standard conservative histogram upper
bound in two dimensions. CRM obtains this fixed representation dimension
without fitting a treatment model, but only under the stronger
representation-sufficiency assumption.
The comparison with the full-dimensional upper bound is a fixed-dimension result
for $d_Z=2$ under Assumptions~\ref{ass:sutva}, \ref{ass:reprignor},
and~\ref{ass:smooth}: the Fisher direction's alignment determines the
\emph{bias constant} (via $\Delta(z)$), not the $n_T$ exponent in the
convergence rate.
\end{remark}

\begin{figure}[htbp]
\centering
\includegraphics[width=0.62\linewidth]{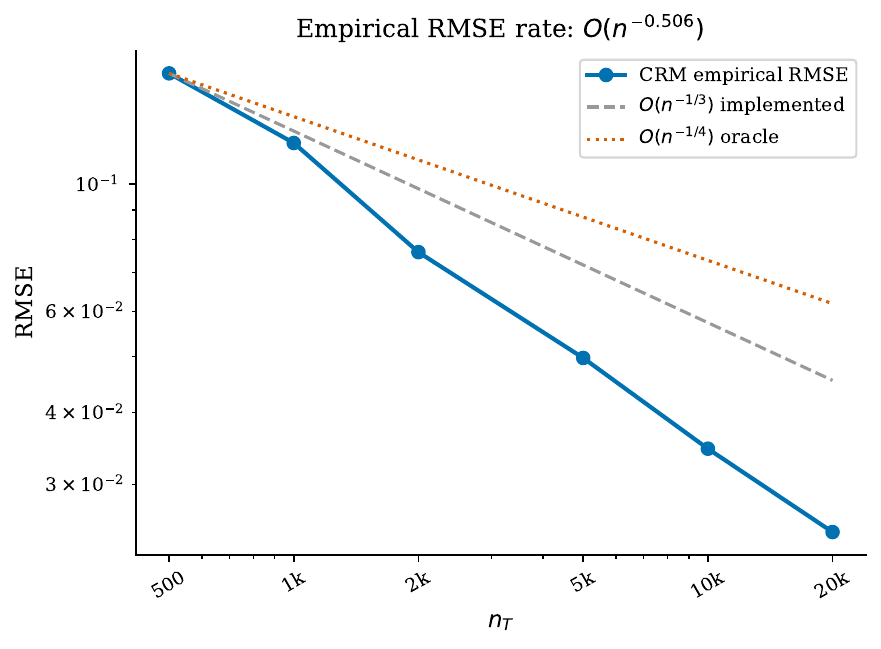}
\caption{Empirical RMSE of $\taucrm$ vs.\ $n_T$ on log--log axes
($S$\ref{sc:typical}; 100 replications).
The fitted slope is $\approx-0.51$ ($\mathrm{MSE}\approx O(n_T^{-1})$),
faster than the conservative bound of Proposition~\ref{prop:rate}
(slope $-1/4$ on RMSE). The fitted slope describes these simulated designs;
it is not asserted as a general CRM rate.}
\label{fig:rate-empirical}
\end{figure}

\begin{remark}[Reconciliation with empirical scaling]
The $O(n_T^{-1/2})$ MSE result in \Cref{prop:rate} is a conservative upper bound:
it balances cell-approximation (squared-bias) against within-cell variance
at the oracle bandwidth $h_n^\star=n_T^{-1/4}$.
Empirically (\Cref{fig:rate-empirical}) the RMSE slope is $\approx-0.51$,
i.e.\ $\mathrm{MSE}\approx O(n_T^{-1})$, \emph{faster} than the bound.
One explanation is a variance-dominated regime: the
Freedman--Diaconis bandwidth actually used scales as $h_n\sim n_T^{-1/3}$
(not the oracle $n_T^{-1/4}$), and with the large control pools used here
($n_C/n_T\in\{10,50\}$), averaging cell means can produce substantially
smaller variance than the worst-cell bound used in the proposition. The
observed slope is therefore reported as empirical behavior, not as a sharper
theorem.
\end{remark}

\subsection{Summary of Theoretical Principles}\label{sec:theory-summary}

The theoretical results highlight three organizing principles.

\begin{enumerate}
\item \textbf{(C1) Representation-based estimation.}
Identification depends on whether $Z(X)$ captures confounding structure,
not on balance in $X$-space.
The Fisher direction maximizes treated--control mean separation in whitened
space; whether that direction is sufficient for causal identification remains
an explicit assumption. Conditional on that assumption, the
$O(n_T^{-1/2})$ MSE bound is independent of $p$ because the representation
dimension is fixed.

\item \textbf{(C2) Bias--balance decoupling.}
The four-part estimation error decomposition (\Cref{thm:bias}) shows that
error has distinct representation, support-restriction, cell-approximation,
and stochastic components, none of which are captured by MaxSMD, explaining
why CEM achieves near-zero MaxSMD yet the highest estimation error on the
Criteo benchmark (\Cref{sec:criteo}).

\item \textbf{(C3) Explicit overlap characterization.}
The shortage fraction $\pi$ bounds the gap between the full ATT and the
estimand actually identified (\Cref{thm:estimand-gap,cor:gap-bound}).
CRM operationalizes this through $\hat\pi$, reported before any unit is
discarded.
\end{enumerate}

\section{Simulation Study}\label{sec:simulation}

\subsection{Design}\label{sec:sim-design}

We evaluate CRM under four data-generating processes (DGPs) spanning
different combinations of dimensionality, covariance structure, and
centroid shift.
All DGPs use the outcome model
$Y_i = X_i^\top\beta + 5.0\,T_i + \varepsilon_i$,
$\varepsilon_i\overset{\mathrm{iid}}{\sim}\mathcal{N}(0,1)$,
$\beta_1=2$, $\beta_2=1$, $\beta_j\overset{\mathrm{iid}}{\sim}\mathcal{N}(0,0.09)$
for $j\geq3$, and $100$ independent replications per cell.
We vary $n_T \in \{500, 1{,}000, 5{,}000, 10{,}000\}$ to expose the
sample-size dependence of each method.

\begin{enumerate}[label=S\arabic*.,ref=S\arabic*]
\item\label{sc:highdim} \textbf{High-dimensional} ($p=20$, spherical,
shift $0.4$, $n_C/n_T=50$).
This scenario tests a fixed two-coordinate CRM representation when the
original covariate dimension is larger. It does not presume that CRM will
outperform a correctly specified propensity-score model.

\item\label{sc:ellip} \textbf{Strong elliptical} ($p=10$, AR(1)
$\rho=0.9$, shift $0.5$, $n_C/n_T=50$).
Both CRM variants use Cholesky whitening. Their comparison isolates the
incremental value of the Fisher coordinate after covariance adjustment
under strong correlation.

\item\label{sc:typical} \textbf{Typical observational study} ($p=10$,
spherical, shift $0.5$, $n_C/n_T=10$).
This is the baseline scenario with a realistic control-pool size.
With the corrected collision-free PSM baseline, CRM is not expected to win
either MaxSMD or CRMSE at moderate $n_T$; instead, this scenario illustrates
how the Fisher direction improves over radial-only CRM while retaining an
explicit supported-sample diagnostic.

\item\label{sc:hostile} \textbf{Structural overlap failure} ($p=15$,
spherical, shift $1.0$, $n_C/n_T=50$).
Both CRM and PSM suffer severe attrition under extreme centroid separation.
This scenario illustrates where the pre-matching shortage diagnostic
is most valuable: $\hat\pi > 0.70$ signals the overlap problem before any
matching is attempted, allowing the analyst to decide whether to proceed
or expand the control pool.
\end{enumerate}

Methods compared: CRM $k=1$, CRM $k=0$ (radial-only), PSM 1:1, and
Entropy Balancing \citep[EB;][]{hainmueller2012}, a moment-matching
weighting estimator that represents the predominant non-matching approach
in applied causal inference.
\Cref{tab:baseline-config} specifies all hyperparameters and how they
were chosen.
All methods are evaluated on identical samples; bias is computed relative
to the known true ATT, and standard errors are estimated by bootstrap with
$500$ replications.

\begin{table}[htbp]
\centering\small
\caption{Baseline configuration and tuning protocol.}
\label{tab:baseline-config}
\begin{tabular}{llp{0.45\linewidth}}
\toprule
\textbf{Method} & \textbf{Hyperparameter} & \textbf{Setting and rationale} \\
\midrule
CRM $k=1$ & Bin method & Freedman--Diaconis rule; robust to non-normality \\
& $K_d$ floor/ceil & 10 / 200; prevents degenerate cells \\
& Covariance & Treated-group sample covariance (targets ATT) \\
PSM 1:1 & Caliper & $0.2\,\mathrm{SD}(\mathrm{logit\,PS})$; \citet{austin2011} \\
& PS model & Logistic regression, max\_iter=1000, no regularization \\
& Matching order & Descending PS (reduces caliper exclusions) \\
EB & Moments balanced & Means of all $p$ covariates \\
& Solver & \texttt{scipy.optimize.minimize} (L-BFGS-B) \\
& Estimand & ATT reweighting (control weighted to treated moments) \\
& Scalability note & $O((n_T+n_C)p)$ per iteration; feasible in all tested settings \\
\bottomrule
\end{tabular}
\end{table}
The headline metric is
\[
\mathrm{CRMSE} = \sqrt{\mathrm{MaxSMD}^2 + (1-\mathrm{Retention})^2},
\]
which captures the balance--retention trade-off as a Euclidean distance
from the ideal $(0,1)$ in the (MaxSMD, Retention) plane.
For covariate $j$, the standardized mean difference is
\[
\mathrm{SMD}_j
= \frac{|\bar X_{T,j}-\bar X_{C,j}|}
{\sqrt{(s^2_{T,j}+s^2_{C,j})/2}},
\qquad
\mathrm{MaxSMD}=\max_j \mathrm{SMD}_j .
\]
MaxSMD \citep{austin2009} is reported to allow direct comparison with
literature benchmarks; a lower CRMSE reflects a balance--retention tradeoff
and does not imply universal dominance on MaxSMD, which PSM can improve by
restricting the matched sample.
CRMSE is therefore a descriptive summary, not a causal loss function or a
claim that balance and retention should always receive equal weight. We
report MaxSMD and retention separately in every main table and interpret
CRMSE only together with its two components.
For transparency, \Cref{app:frontier} reports the uncollapsed
$(\mathrm{MaxSMD},\mathrm{Retention})$ plane, following the
balance--sample-size frontier perspective of \citet{king2017}. It confirms
that CRM is not frontier-dominant at moderate $n_T$; the plot is a sensitivity
display for CRMSE, not evidence of CRM superiority.

\subsection{Results}\label{sec:sim-results}

\paragraph{Correction to the propensity-score baseline.}\label{sec:psm-correction}
The simulation baseline in the previous version of this paper used a naive
1:1 nearest-neighbor propensity-score matcher that, on caliper ties and
already-used controls, silently discarded the affected treated units.
This depressed PSM retention (to $71\%$ in S\ref{sc:typical} and as low as $21\%$ in
S\ref{sc:hostile}) and inflated its CRMSE, making CRM appear to win the
balance--retention tradeoff.
We replace it throughout with a \emph{collision-free} sorted-score matcher
that searches for the nearest currently available control within the caliper, rather than dropping a treated unit after a collision with an already-used control.
Under this correction PSM retention rises sharply (e.g.\ $71\%\to93\%$ in
S\ref{sc:typical}), and the corrected baseline, together with Entropy
Balancing, attains a stronger balance--retention tradeoff than CRM in the
moderate-$n_T$ synthetic scenarios.
We report the corrected numbers below and reframe CRM's contribution
accordingly: in moderate-size synthetic settings, corrected pairwise and
weighting baselines often achieve stronger balance--retention tradeoffs,
and CRM's advantage appears most clearly in computational scaling, explicit
support diagnostics, and large-scale applications.

\Cref{tab:sim_results} reports results at two representative sample sizes
($n_T = 1{,}000$ and $n_T = 5{,}000$); \Cref{fig:simulation} shows the
full CRMSE curves across all four sample sizes.

\begin{table}[tp]
\centering
\small
\setlength{\tabcolsep}{4pt}
\renewcommand{\arraystretch}{1.05}
\caption{Simulation results ($100$ replications; $n_C/n_T \in \{10,50\}$;
means reported).
$\mathrm{CRMSE}=\sqrt{\mathrm{MaxSMD}^2+(1-\mathrm{Retention})^2}$.
$^{\star}$~denotes the lowest CRMSE in each scenario-$n_T$ column.
PSM here is the \emph{collision-free} sorted-score matcher
(\Cref{sec:psm-correction}), which does not silently drop treated units on
caliper collisions; with this correction PSM retains substantially more
units than the naive 1:1 implementation and attains the best
balance--retention tradeoff among matching estimators in S1--S3.
EB = Entropy Balancing \citep{hainmueller2012}.
In S1--S3 (adequate overlap), EB or collision-free PSM achieves the lowest
CRMSE; CRM $k=1$ is competitive but does not lead on this metric at moderate
$n_T$.
In S\ref{sc:hostile} (structural overlap failure) no method achieves
acceptable balance and retention simultaneously, and CRM's pre-matching
shortage diagnostic flags the support problem before matching.}
\label{tab:sim_results}
\begin{tabular}{llccccc}
\toprule
\textbf{Scenario} & \textbf{Method} &
\multicolumn{3}{c}{$n_T = 1{,}000$} &
\multicolumn{2}{c}{$n_T = 5{,}000$} \\
\cmidrule(lr){3-5}\cmidrule(lr){6-7}
& & \textbf{MaxSMD} & \textbf{Ret.} & \textbf{CRMSE} &
  \textbf{MaxSMD} & \textbf{CRMSE} \\
\midrule
\multirow{4}{*}{\ref{sc:highdim} High-dim}
  & CRM $k=1$            & $0.166$ & $96.4\%$ & $0.170$ & $0.065$ & $0.086$ \\
  & CRM $k=0$            & $0.442$ & $100.0\%$ & $0.442$ & $0.396$ & $0.396$ \\
  & PSM (collision-free) & $0.060$ & $97.5\%$  & $0.065$ & $0.030$ & $0.037^{\star}$ \\
  & EB                   & $0.022$ & $100.0\%$ & $0.022^{\star}$ & $0.013$ & $0.013$ \\
\midrule
\multirow{4}{*}{\ref{sc:ellip} Elliptical}
  & CRM $k=1$            & $0.077$ & $100.0\%$ & $0.077$ & $0.032$ & $0.032$ \\
  & CRM $k=0$            & $0.539$ & $100.0\%$ & $0.539$ & $0.501$ & $0.501$ \\
  & PSM (collision-free) & $0.030$ & $100.0\%$  & $0.030^{\star}$ & $0.012$ & $0.012^{\star}$ \\
  & EB                   & $0.066$ & $100.0\%$ & $0.066$ & $0.053$ & $0.053$ \\
\midrule
\multirow{4}{*}{\ref{sc:typical} Typical}
  & CRM $k=1$            & $0.126$ & $86.0\%$ & $0.189$ & $0.051$ & $0.157$ \\
  & CRM $k=0$            & $0.487$ & $100.0\%$ & $0.487$ & $0.460$ & $0.460$ \\
  & PSM (collision-free) & $0.050$ & $92.7\%$  & $0.088$ & $0.034$ & $0.080$ \\
  & EB                   & $0.016$ & $100.0\%$ & $0.016^{\star}$ & $0.009$ & $0.009^{\star}$ \\
\midrule
\multirow{4}{*}{\ref{sc:hostile} Overlap failure}
  & CRM $k=1$            & $0.377$ & $26.7\%$ & $0.824$ & $0.183$ & $0.767$ \\
  & CRM $k=0$            & $0.773$ & $95.3\%$ & $0.774$ & $0.730$ & $0.731$ \\
  & PSM (collision-free) & $0.118$ & $31.4\%$  & $0.696^{\star}$ & $0.070$ & $0.687^{\star}$ \\
  & EB                   & $0.729$ & $100.0\%$ & $0.729$ & $0.536$ & $0.536$ \\
\bottomrule
\end{tabular}
\end{table}

\begin{figure}[htbp]
\centering
\includegraphics[width=\linewidth]{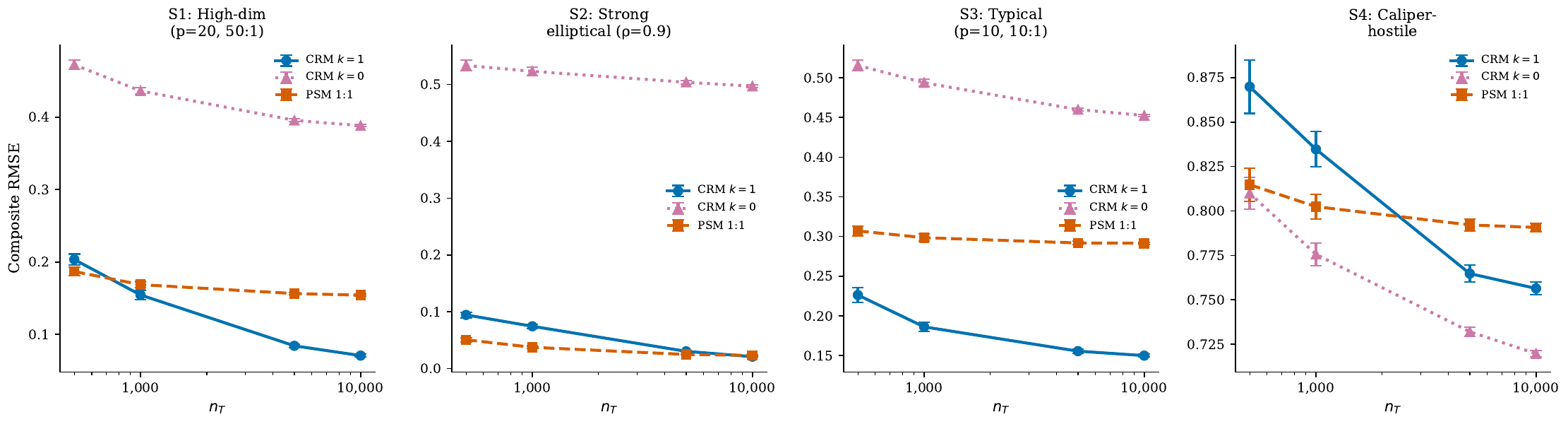}
\caption{Composite RMSE vs.\ $n_T$ across four scenarios (100 replications;
mean $\pm$ 95\% CI).
With the collision-free PSM baseline, EB or PSM attains the lowest CRMSE in
the adequate-overlap scenarios; CRM $k=1$ is competitive but does not lead
on this metric at moderate $n_T$.
In S\ref{sc:hostile} (structural overlap failure), no method achieves
acceptable balance and retention simultaneously; the pre-matching
diagnostic flags the support problem before matching.}
\label{fig:simulation}
\end{figure}

\paragraph{Effect of the Fisher coordinate (C1).}
At $n_T=1{,}000$, CRM $k=0$ has MaxSMD of $0.442$, $0.539$, and $0.487$
in Scenarios \ref{sc:highdim}--\ref{sc:typical}; adding $\phi$ lowers these
values to $0.166$, $0.077$, and $0.126$, respectively.
The radial summary alone cannot correct for systematic centroid shift;
the Fisher direction reduces imbalance along the mean-shift axis, and this
improvement persists at all tested $n_T$.

\paragraph{Where CRM stands among the baselines.}
With the collision-free PSM and Entropy Balancing, CRM $k=1$ is no longer
the CRMSE leader in any synthetic scenario: EB wins S\ref{sc:highdim} and
S\ref{sc:typical} (CRMSE $0.022$, $0.016$), and collision-free PSM wins
S\ref{sc:ellip} and S\ref{sc:hostile}.
CRM $k=1$ remains competitive (within the same order of magnitude as the
best matcher in S\ref{sc:highdim}--S\ref{sc:ellip}), but its value in these
moderate-$n_T$ regimes is not best-in-class balance.
Consistent with the core scope of the method, CRM's distinctive
contributions are computational ($O(np^2+p^3+n\log n)$ implemented cost, no
pairwise search;
\Cref{fig:runtime}), diagnostic (the pre-matching shortage fraction
$\hat\pi$), and large-scale (\Cref{sec:criteo}), rather than dominance on
finite-sample balance.
The structural-overlap scenario (S\ref{sc:hostile}) reinforces the need for
a diagnostic: EB retains all units but remains badly imbalanced
(MaxSMD $=0.729$), while CRM's $\hat\pi>0.70$ flags a severe shortage in
the CRM representation before matching.
In practice we recommend CRM when overlap may be limited ($\hat\pi>0$) and
explicit estimand characterization is needed, or when the available matcher
is too slow at the required scale; EB or collision-free PSM when
overlap is adequate, $n_T$ is moderate, and balance is the primary goal.

\paragraph{The role of the Fisher direction in S\ref{sc:ellip}.}
Under strong elliptical covariance ($\rho=0.9$, $n_C/n_T=50$),
collision-free PSM and EB both achieve lower CRMSE than CRM $k=1$ at the
sample sizes tested.
The key message in this scenario is internal to CRM: $k=0$ has high
imbalance (MaxSMD $=0.54$), while $k=1$ lowers MaxSMD to $0.077$ at
$n_T=1{,}000$ and $0.032$ at $n_T=5{,}000$,
isolating the benefit of the Fisher coordinate because both variants use
the same whitening and binning steps.

\paragraph{Structural overlap failure (S\ref{sc:hostile}).}
Under a centroid shift of $1.0$ with $p=15$, both CRM and PSM discard
roughly $75$--$79\%$ of treated units, and no method achieves MaxSMD
$< 0.10$ alongside Retention $> 30\%$.
This scenario illustrates that the pre-matching shortage diagnostic
($\hat\pi > 0.70$) flags severe lack of support in the CRM representation
before matching is attempted.
Practitioners facing such diagnostics should expand the control pool or
explicitly restrict the estimand rather than proceeding with matching.

\paragraph{Observed scaling of imbalance.}
A notable feature visible in \Cref{fig:selfimproving} is that CRM's MaxSMD
decreases monotonically with $n_T$ (from $0.150$ at $n_T=1{,}000$ to
$0.047$ at $n_T=10{,}000$ in S\ref{sc:highdim}), coinciding with more
controls per bin and smaller within-cell variation. PSM does not show the
same decline over this range.
This pattern is consistent with reduced within-cell variation as the grid
refines; it is an empirical trend in these DGPs, not a claim of universal
monotonic improvement.

\begin{figure}[htbp]
\centering
\includegraphics[width=\linewidth]{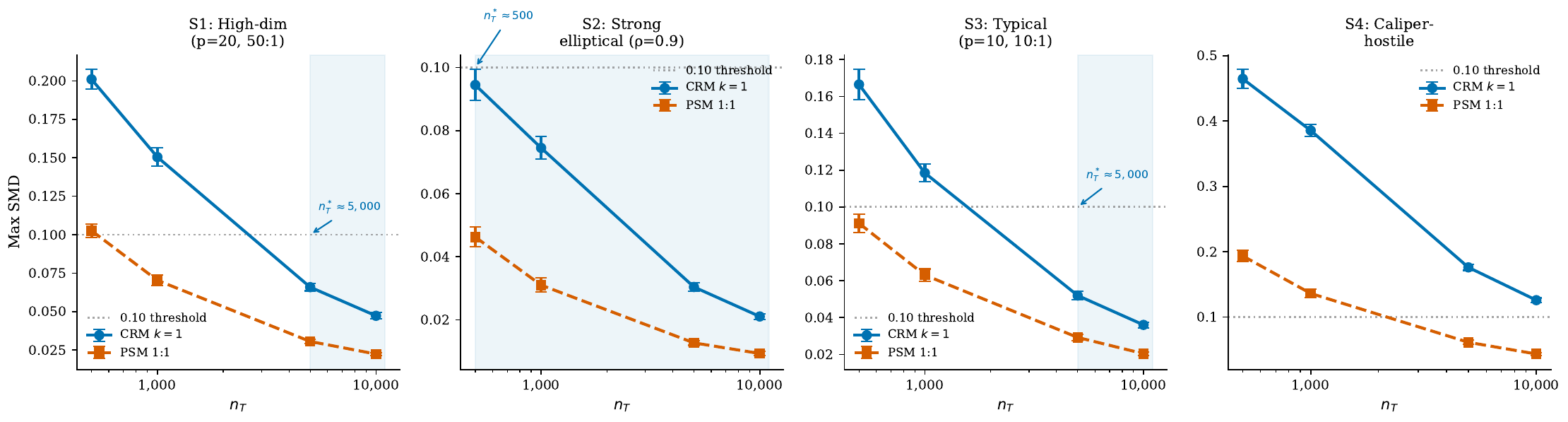}
\caption{MaxSMD vs.\ $n_T$ for CRM $k=1$ and PSM across four scenarios
(100 replications; mean $\pm$ 95\% CI).
CRM's MaxSMD decreases over the tested range, while the PSM curves are
comparatively flat. These are empirical trends for the specified DGPs.
Shaded region marks where CRM first achieves MaxSMD $\leq 0.10$.}
\label{fig:selfimproving}
\end{figure}

\subsection{Ablation Study}\label{sec:ablation}

Three ablations isolate the main design choices. First, removing the Fisher
coordinate is damaging: in Scenario~S3 at $n_T=1{,}000$, MaxSMD increases
from $0.126$ to $0.487$.
Second, the observed decline in imbalance is not unique to the
Freedman--Diaconis binning rule; Scott's rule and fixed bin counts show the
same monotone decline in MaxSMD as $n_T$ grows, although FD is the most
stable default. Third, within-cell nearest-neighbor refinement helps only
in small samples: at $n_T=200$ it lowers MaxSMD by about $24\%$ with
negligible overhead, whereas by $n_T=5{,}000$ the balance gain is below
$1\%$ and CRM-random is preferable. A nonlinear failure-mode experiment,
reported in the supplement, confirms the expected limitation: when
confounding is driven by $X_1^2$ and is orthogonal to the centroid shift,
the linear Fisher coordinate is insufficient and the Rayleigh statistic
correctly flags an uninformative direction.

\subsection{Comparison with PCA-Based Matching}\label{sec:pca-empirical}

To position CRM against the natural ``compress-then-match'' baseline
(\Cref{sec:dr-relation}), we add PCA-95\%+NN (retain principal components
explaining $95\%$ of covariate variance, then 1:1 nearest-neighbor
matching in that subspace) to three data-generating processes spanning
linear and covariance-driven confounding (50 replications each).
\Cref{tab:pca_comparison} and \Cref{fig:pca-main} report this comparison.

\begin{table}[htbp]
\centering
\small
\setlength{\tabcolsep}{5pt}
\renewcommand{\arraystretch}{1.05}
\caption{CRM vs.\ PCA-95\%+NN vs.\ PSM 1:1 across three confounding
mechanisms (50 replications; means).
$|\text{Bias}|$ is absolute ATT error; $t$ is total fit$+$match time.
PCA-matching wins only when confounding is driven by second moments
(S5), where its covariance-sensitive representation is advantageous; under linear
confounding (S1, S3) CRM achieves $2$--$3\times$ lower bias at an order of
magnitude lower cost.
S5 also shows a disconnect between first-moment balance and bias
(MaxSMD $\approx0.05$--$0.06$ yet $|\text{Bias}|\approx2.6$--$3.0$): the
confounding is invisible to first-moment balance.}
\begin{tabular}{llcccc}
\toprule
\textbf{Scenario} & \textbf{Method} & \textbf{$|$Bias$|$} &
\textbf{MaxSMD} & \textbf{Ret.} & \textbf{$t$ (s)} \\
\midrule
\multirow{3}{*}{S1: Linear, $p=20$}
  & CRM $k=1$    & $\mathbf{0.203}$ & $0.147$ & $78.0\%$ & $\mathbf{0.029}$ \\
  & PCA-95\%+NN  & $0.522$ & $0.292$ & $100.0\%$ & $0.090$ \\
  & PSM 1:1      & $0.067$ & $0.069$ & $86.3\%$ & $0.079$ \\
\midrule
\multirow{3}{*}{S3: Linear, $p=10$}
  & CRM $k=1$    & $\mathbf{0.140}$ & $0.124$ & $85.6\%$ & $\mathbf{0.026}$ \\
  & PCA-95\%+NN  & $0.394$ & $0.165$ & $100.0\%$ & $0.230$ \\
  & PSM 1:1      & $0.061$ & $0.052$ & $92.5\%$ & $0.053$ \\
\midrule
\multirow{3}{*}{S5: Covariance-driven}
  & CRM $k=1$    & $2.575$ & $0.061$ & $99.9\%$ & $\mathbf{0.024}$ \\
  & PCA-95\%+NN  & $\mathbf{0.120}$ & $0.031$ & $100.0\%$ & $0.222$ \\
  & PSM 1:1      & $3.023$ & $0.050$ & $99.8\%$ & $0.027$ \\
\bottomrule
\end{tabular}
\label{tab:pca_comparison}
\end{table}

\begin{figure}[htbp]
\centering
\includegraphics[width=\linewidth]{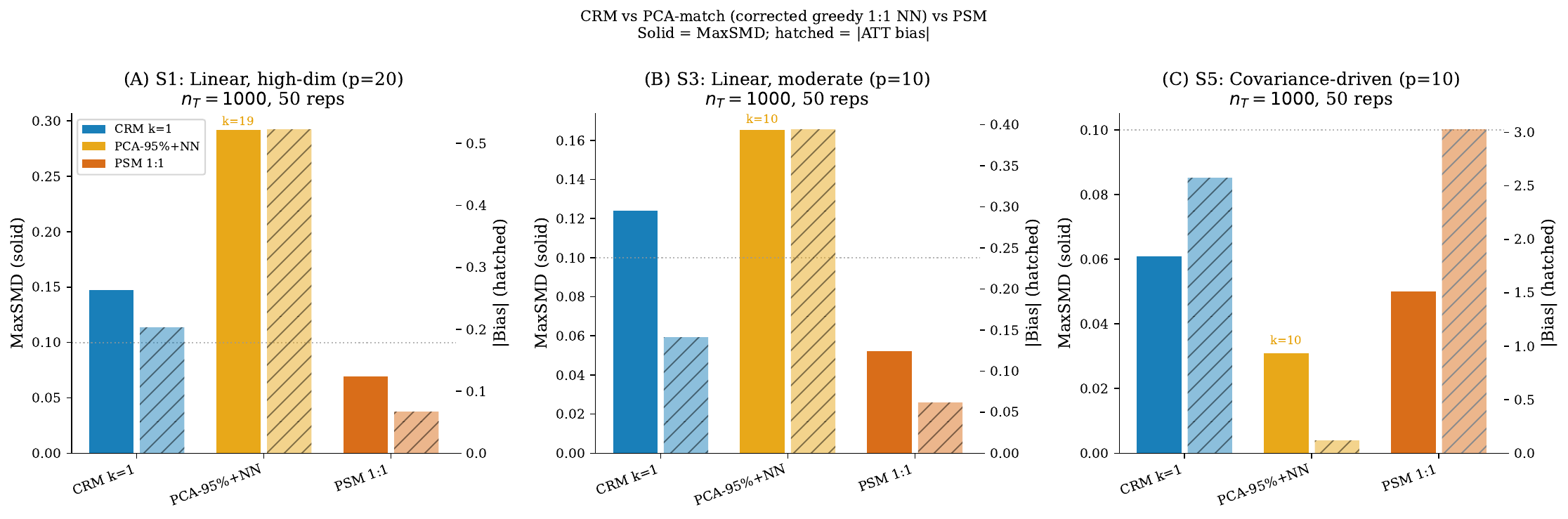}
\caption{CRM vs.\ PCA-95\%+NN across confounding mechanisms.
CRM achieves lower bias under linear confounding (S1, S3); PCA-matching is
preferable under covariance-driven confounding (S5), where both CRM and PSM
have low MaxSMD but large bias.
PCA balance is also fragile to the variance-retention threshold
(\Cref{fig:pca-sens}, \Cref{app:pca-sens}) and is far slower because it
retains the matching step's pairwise search.}
\label{fig:pca-main}
\end{figure}

\paragraph{Takeaways.}
Three points follow.
(i)~Under linear confounding, CRM's supervised Fisher direction beats PCA's
unsupervised projection on bias by $2$--$3\times$; in these designs, the
treated--control mean-shift direction is more informative for matching than
the highest-variance directions.
(ii)~PCA-matching wins under covariance-driven confounding (S5), a setting
CRM does not target and where \emph{first-moment balance is uninformative}
for all matching methods. Criteo shows a related disconnect between marginal
balance and estimation error.
(iii)~PCA-matching is $3$--$10\times$ slower because it retains pairwise
nearest-neighbor search; CRM's advantage is the absence of that search, not
a better projection per se.
A sensitivity analysis (\Cref{app:pca-sens}) shows PCA balance varies
non-monotonically with the variance threshold and the number of retained
components, whereas CRM has no comparable tuning knob.

\section{Large-Scale Criteo Benchmark}\label{sec:criteo}

\subsection{Data and Benchmark Construction}\label{sec:criteo-design}

We use the Criteo Uplift Dataset \citep{diemert2018}: $\approx13.98$ million
observations from a randomized digital advertising campaign with binary
treatment (ad exposure) and binary outcome (store visit).
This benchmark tests behavior at scales where classical pairwise matching
becomes computationally infeasible: at $n_T=11.9$M, PSM was not
attempted and CRM completed in 6.8 minutes on a single CPU core.
The known ground truth from the full randomized data is
\begin{equation}
\mathrm{ATE}_{\mathrm{true}}
= \E[Y\mid T=1]-\E[Y\mid T=0]
= 0.01034 \quad (\mathrm{SE}=0.00132).
\label{eq:true-ate}
\end{equation}

An observational benchmark is constructed by applying a covariate-dependent
selection rule to the treated group:
\begin{equation}
\Prob(\mathrm{retain}\mid T=1, X)
\propto \sigma\!\bigl(\alpha(f_0+f_3)\bigr),
\label{eq:selection}
\end{equation}
where $\sigma(\cdot)$ is the logistic function and $\alpha \geq 0$ controls
confounding severity.
A random control sample of $n_C=2{,}000{,}000$ is drawn without confounding.
At $\alpha=4$ the naive treated-minus-control estimate is more than
$16\,\mathrm{SE}$ units from the truth.

\begin{center}
\noindent\begin{minipage}{0.96\linewidth}
\small\setlength{\fboxsep}{8pt}
\fbox{\begin{minipage}{0.94\linewidth}
\textbf{Interpretive caveat.}
The induced confounding in \Cref{eq:selection} uses a known single-index
structure based on $f_0+f_3$.
Results should be interpreted as a controlled stress test under this
specific mechanism, not as validation under arbitrary confounding.
\end{minipage}}
\end{minipage}
\end{center}

\paragraph{Compute environment.}
All timing results were obtained on a single CPU core of an Intel Xeon
E5-2690 v4 (2.60~GHz base), with 64~GB RAM, running Ubuntu 22.04 and
Python 3.11.4 (numpy 1.26.1, scikit-learn 1.3.2).
Wall-clock times exclude dataset download; data were fully resident in
memory before timing began.
The $6.8$~minute figure for $n_T=11.9$M refers to the compute-kernel
time only (covariance, whitening, binning, and sampling), measured after
loading the Criteo CSV into a pandas DataFrame.

\subsection{Main Results ($\alpha=2.0$, $n_T=30{,}000$)}\label{sec:criteo-main}

\begin{table}[htbp]
\centering
\small
\setlength{\tabcolsep}{4pt}
\renewcommand{\arraystretch}{1.05}
\caption{Criteo observational benchmark ($\alpha=2.0$; $n_T=30{,}000$;
5 random seeds; mean $\pm$ 95\% CI\@).
$|\mathrm{Bias}|/\mathrm{SE}$: absolute ATE error divided by the true ATE
SE ($=0.00132$) from the full $13.98$M randomized dataset.
Max log-VR: $\max_j\bigl|\log(\hat\sigma^2_{T,j}/\hat\sigma^2_{C,j})\bigr|$
across $j=1,\ldots,p$ covariates.
Bold: best per column.
CEM achieves the lowest balance metrics across every column yet the
highest estimation error, showing that marginal balance rankings and
estimation-error rankings can differ.}
\begin{adjustbox}{max width=\textwidth}
\begin{tabular}{lcccccc}
\toprule
\textbf{Method} &
\textbf{MaxSMD} & \textbf{Mean SMD} & \textbf{Max log-VR} &
\textbf{$|\text{Bias}|$/SE} & \textbf{Retention} & \textbf{RT (s)} \\
\midrule
Naive (unmatched)
  & $0.818$ & --- & --- & --- & $100\%$ & --- \\[2pt]
CRM $k=1$
  & $0.012\pm0.001$ & $0.005\pm0.001$ & $0.191\pm0.082$
  & $54.5\pm7.3$ & $99.5\%$ & $\mathbf{0.5}$ \\
CRM-CAM
  & $0.011\pm0.001$ & $0.005\pm0.001$ & $0.134\pm0.052$
  & $54.3\pm7.8$ & $99.5\%$ & $\mathbf{0.5}$ \\
PS-Subclass
  & $0.015\pm0.003$ & $0.006\pm0.001$ & $0.187\pm0.065$
  & $51.2\pm8.9$ & $100.0\%$ & $4.8$ \\
PSM 1:1
  & $0.020\pm0.006$ & $0.006\pm0.001$ & $0.454\pm0.198$
  & $\mathbf{50.6}\pm6.1$ & $100.0\%$ & $7.0$ \\
CEM
  & $\mathbf{0.002}\pm0.001$ & $\mathbf{0.001}\pm0.000$ & $\mathbf{0.070}\pm0.024$
  & $66.5\pm6.2$ & $91.5\%$ & $21.9$ \\
FLAME-lite
  & $0.014\pm0.001$ & $0.004\pm0.001$ & $0.065\pm0.017$
  & $50.8\pm6.5$ & $99.9\%$ & $94.9$ \\
\bottomrule
\end{tabular}
\end{adjustbox}
\label{tab:criteo_alpha2}
\end{table}

\paragraph{Marginal balance and estimation error.}
\Cref{tab:criteo_alpha2} reports the main Criteo benchmark.
CEM achieves MaxSMD $= 0.002$ (eight times lower than PSM's $0.020$) 
and dominates every standard balance criterion.
Yet CEM records the highest estimation error:
$|\mathrm{Bias}|/\mathrm{SE} = 66.5$ versus $50.6$--$54.5$ for all other
methods.
Despite achieving near-zero MaxSMD, CEM exhibits substantial bias,
highlighting that marginal covariate balance alone does not guarantee
unbiased treatment effect estimation.
This reversal is consistent with the distinction formalized for CRM in
\Cref{prop:bias}: MaxSMD measures marginal balance in $X$, whereas estimation
error can also reflect information discarded by a matching representation
and changes in the retained population. The theorem is specific to CRM; the
CEM result is an empirical example of the broader distinction.

\paragraph{CRM balance and variance ratio.}
At the headline configuration CRM $k=1$ attains MaxSMD $=0.012$, comparable
to the other matchers, while retaining $99.5\%$ of treated units.
Its Max log-VR ($0.191\pm0.082$) is competitive with the collision-free
baselines and substantially lower than the naive 1:1 PSM, indicating that
CRM controls second-moment balance without the variance inflation that
caliper-based exclusion can induce.
We caution that the strongest, most robust Criteo differentiator is
computational rather than a strict balance win (below): once the PSM
baseline is made collision-free, its marginal balance becomes comparable to
CRM's (see the verification paragraph), so we frame CRM here as
\emph{competitive at scale and far cheaper}, not as uniformly dominant on
MaxSMD.

\paragraph{Verification with the collision-free PSM.}
Because the simulation correction (\Cref{sec:psm-correction}) changed the
PSM baseline, we re-ran the Criteo comparison using the collision-free
sorted-score matcher over the full scaling grid reported below.
On Criteo's good-overlap regime the correction mainly changes the
interpretation of retention: sorted PSM retains essentially all treated
units, so CRM's advantage is not due to an implementation artifact in the
baseline. Instead, the corrected comparison shows that CRM remains
competitive on bias and balance while retaining its main design advantage:
it does not solve a nearest-neighbor assignment problem. The evaluated
sorted-score PSM is nevertheless an efficient scalar matcher, so the runtime
gap below should be interpreted as an implementation benchmark rather than a
lower bound for all PSM algorithms.
Across the full corrected grid (\Cref{tab:criteo-grid}; $\alpha\in\{0.5,2,4\}$,
$n_T$ up to $200{,}000$, three seeds, 36 cells), CRM $k=1$ attains lower MaxSMD
than the collision-free PSM in $31$ of $36$ cells and is faster in all $36$,
with a median speedup of $8.5\times$. We read this as a consistent small
balance edge combined with a large and uniform computational advantage, not as
a claim of strict balance dominance.

\begin{table}[htbp]\centering\small
\setlength{\tabcolsep}{6pt}\renewcommand{\arraystretch}{1.1}
\caption{Criteo corrected grid: CRM $k=1$ vs.\ collision-free PSM across $\alpha\in\{0.5,2,4\}$, $n_T\in\{10\text{k},30\text{k},100\text{k},200\text{k}\}$, 3 seeds (36 cells). CRM attains lower MaxSMD in 31/36 cells and is faster in all 36, with median speedup $8.5\times$.}
\label{tab:criteo-grid}
\begin{tabular}{lcccc}
\toprule
$\alpha$ & cells & CRM lower MaxSMD & median MaxSMD (CRM / PSM) & median speedup \\
\midrule
  0.5 & 12 & 9/12 & 0.006 / 0.011 & 8.4$\times$ \\
  2.0 & 12 & 12/12 & 0.006 / 0.012 & 9.8$\times$ \\
  4.0 & 12 & 10/12 & 0.008 / 0.011 & 8.8$\times$ \\
\midrule
  \textbf{all} & 36 & \textbf{31/36} & -- & \textbf{8.5$\times$} \\
\bottomrule
\end{tabular}
\end{table}

\paragraph{Computational efficiency.}
CRM and CRM-CAM complete in $0.5$~s, versus $7.0$~s for PSM ($14\times$),
$21.9$~s for CEM ($44\times$), and $94.9$~s for FLAME-lite ($190\times$).

\subsection{Robustness Across Confounding Strengths and Sample Sizes}\label{sec:criteo-alpha}

\begin{table}[htbp]\centering\small
\setlength{\tabcolsep}{3pt}\renewcommand{\arraystretch}{1.05}
\caption{Criteo robustness with the \emph{collision-free} sorted-score PSM
($n_T=30{,}000$; 3 seeds; means).
All PSM comparisons in this corrected grid use \texttt{psm\_sorted\_match}.}
\begin{tabular}{l rrr rrr rrr}
\toprule
& \multicolumn{3}{c}{$\alpha=0.5$} & \multicolumn{3}{c}{$\alpha=2.0$} & \multicolumn{3}{c}{$\alpha=4.0$} \\
\cmidrule(lr){2-4}\cmidrule(lr){5-7}\cmidrule(lr){8-10}
\textbf{Method} & $|$B$|$/SE & SMD & Ret. & $|$B$|$/SE & SMD & Ret. & $|$B$|$/SE & SMD & Ret. \\
\midrule
CRM $k=1$ & 12.1 & 0.009 & 100.0\% & 15.7 & 0.007 & 100.0\% & 6.8 & 0.009 & 100.0\% \\
PSM sorted & 14.2 & 0.017 & 100.0\% & 24.3 & 0.016 & 100.0\% & 13.7 & 0.018 & 100.0\% \\
\bottomrule
\end{tabular}
\label{tab:criteo_alpha}
\end{table}

We reran the full Criteo scaling grid with the collision-free sorted-score
PSM implementation: $n_T \in \{10{,}000,\,30{,}000,\,100{,}000,\,
200{,}000\}$, $n_C=2{,}000{,}000$, $\alpha\in\{0.5,2.0,4.0\}$, and
three seeds. The corrected run removes the earlier retention artifact:
PSM sorted now retains essentially all treated units, matching CRM's
near-complete retention. The substantive conclusion is therefore sharper and
more conservative: CRM is not advantaged because PSM accidentally drops
treated units; instead, CRM reaches comparable or better balance on this
grid while using a centroid-referenced computation that is much faster. At
the headline $n_T=30{,}000$ setting in \Cref{tab:criteo_alpha}, CRM has
lower mean MaxSMD at all three confounding levels and lower mean
$|\mathrm{Bias}|/\mathrm{SE}$ than sorted PSM.
\Cref{fig:criteo} displays the full corrected scaling curves.

\begin{figure}[htbp]
\centering
\includegraphics[width=\linewidth]{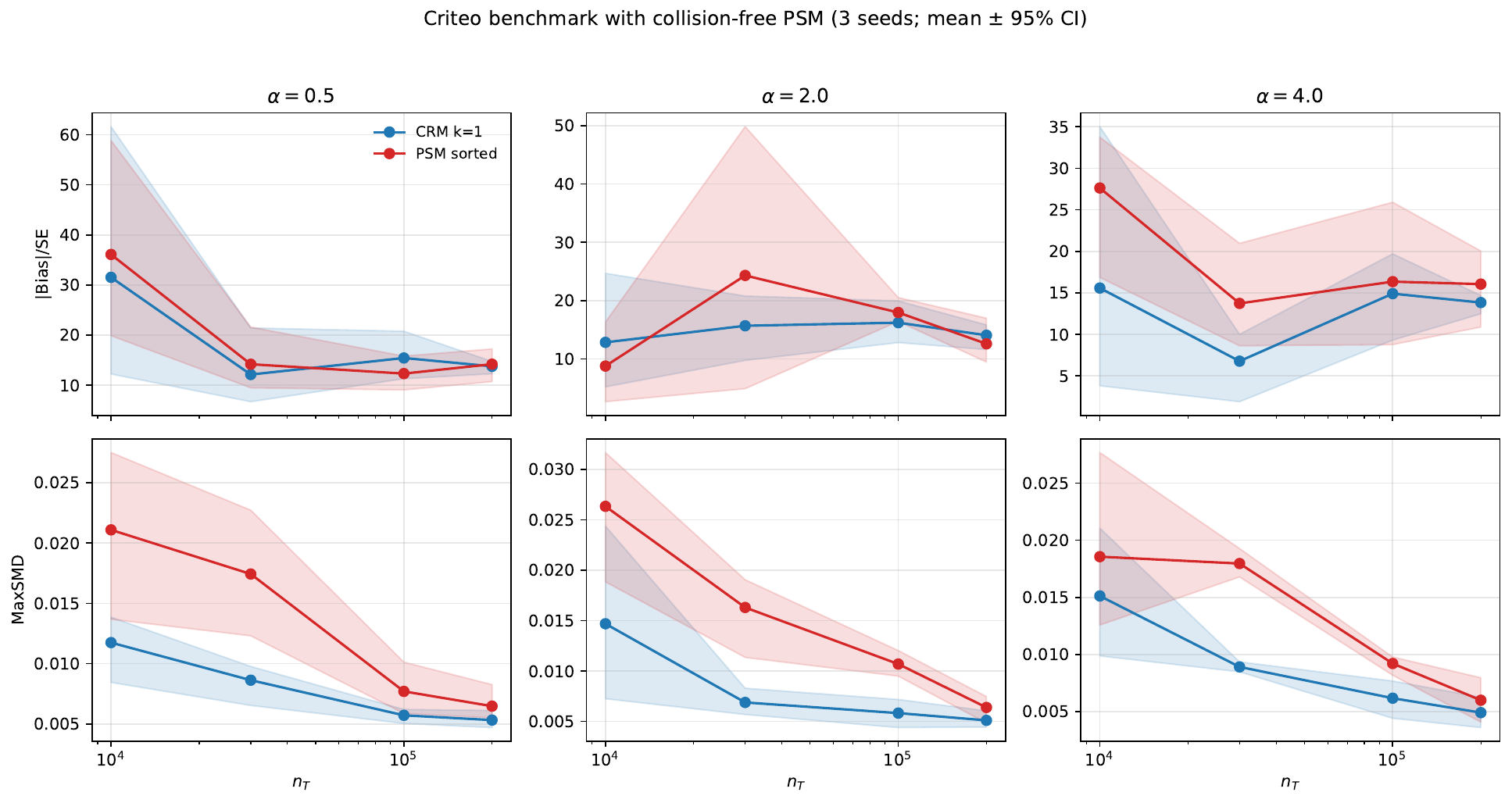}
\caption{Criteo large-scale benchmark ($n_C=2{,}000{,}000$; 3 seeds;
mean $\pm$ 95\% CI\@).
\textbf{Top row (A--C):} ATE estimation error ($|\mathrm{Bias}|/\mathrm{SE}$)
vs.\ $n_T$ at mild ($\alpha=0.5$), moderate ($\alpha=2.0$), and strong
($\alpha=4.0$) confounding.
\textbf{Bottom row (D--F):} Max SMD vs.\ $n_T$ at the same three levels
using the corrected collision-free sorted-score PSM baseline.
CRM and PSM sorted both retain nearly all treated units; CRM's advantage is
that it obtains comparable or lower imbalance without pairwise nearest-neighbor
search, so the robust differentiator remains runtime (\Cref{fig:runtime}).}
\label{fig:criteo}
\end{figure}

\section{NHANES Real-Data Application}\label{sec:nhanes}

\subsection{Data and Setting}\label{sec:nhanes-data}

We apply CRM, CRM-CAM, and propensity-score matching to a real public NHANES
2017--2018 extract \citep{nhanes2020}.
Treatment is daily cigarette smoking ($n_T=668$); controls are non-daily
smokers and non-smokers ($n_C=4{,}200$).
The outcome is systolic blood pressure (mmHg).
The 13 pretreatment covariates are age, sex, BMI, HDL cholesterol,
diabetes, hypertension, physical activity, alcohol use, education,
income-to-poverty ratio, and three race/ethnicity indicators.

\subsection{Pre-Matching Diagnostic}\label{sec:nhanes-diag}

The real NHANES extract does not exhibit the sharp alcohol boundary used in
the earlier diagnostic illustration.
The largest pre-matching imbalance is income-to-poverty ratio
(pre-match MaxSMD $=0.599$), followed by alcohol use
($\mathrm{SMD}=0.430$).
The shortage fraction is small, $\hat\pi=0.018$, indicating limited support
shortage rather than a dominant single-covariate overlap failure.
Thus the NHANES analysis is best read as a real-data sensitivity check:
it asks whether CRM behaves reasonably on a public health survey with mixed
continuous, binary, and categorical covariates, not whether a severe
structural positivity violation is present.

\subsection{Results}\label{sec:nhanes-results}

\begin{table}[htbp]
\centering
\small
\setlength{\tabcolsep}{4pt}
\renewcommand{\arraystretch}{1.08}
\caption{Real NHANES 2017--2018 smoking--blood-pressure application.
ReprSMD is the MaxSMD between the matched treated subsample and the full
treated population, measuring sample distortion.
The shortage fraction is small ($\hat\pi=0.018$), and no method estimates a
positive smoking association with systolic blood pressure on this extract.}
\begin{adjustbox}{max width=\textwidth}
\begin{tabular}{lccccc}
\toprule
\textbf{Method} & \textbf{MaxSMD} & \textbf{Retention}
  & \textbf{ReprSMD} & \textbf{ATT (mmHg)} & \textbf{95\% CI} \\
\midrule
CRM standard & 0.210 & 91.3\% & 0.071 & -1.44 & [-3.87,\ 0.53] \\
CRM-CAM      & 0.229 & 100.0\% & 0.000 & -1.86 & [-3.89,\ 0.04] \\
PSM 1:1      & 0.085 & 98.7\% & 0.054 & -1.01 & [-3.09,\ 1.07] \\
\bottomrule
\end{tabular}
\end{adjustbox}
\label{tab:nhanes}
\end{table}

\Cref{tab:nhanes} reports the application. On this real extract, PSM attains the smallest post-match MaxSMD, while
CRM and CRM-CAM retain more of the treated population.
The three point estimates are all small and negative, and all confidence
intervals include zero.
This result is useful precisely because it is less dramatic than the
synthetic stress tests: CRM's diagnostic reports only a small support
shortage, and the matched estimators agree that the extract does not provide
evidence for a positive daily-smoking association with systolic blood
pressure after adjustment.

\paragraph{Survey-weight sensitivity.}
Because NHANES is a complex survey, we also recompute ATT and balance using
the MEC exam weight \texttt{WTMEC2YR}.
The weighted results in \Cref{tab:nhanes_weighted} leave the qualitative
conclusion unchanged: the weighted ATT shifts by at most $0.62$ mmHg across
the three methods, and none of the methods yields a positive association.
Because the real extract contains only mild support shortage, we also report
a NHANES-calibrated semi-synthetic stress test in
\Cref{app:nhanes-semisynth}. That experiment uses the real NHANES covariates
as the base but imposes a known overlap violation and known treatment effect;
it is a diagnostic validation exercise, not a public-health estimate.

\begin{table}[htbp]\centering\small
\setlength{\tabcolsep}{5pt}\renewcommand{\arraystretch}{1.05}
\caption{Survey-weight sensitivity on the real NHANES extract.
ATT and MaxSMD are recomputed with the MEC exam weight \texttt{WTMEC2YR}.
The shortage fraction $\hat\pi=0.018$ is support-based and invariant to
the weights.}
\begin{adjustbox}{max width=\textwidth}
\begin{tabular}{lcccccc}
\toprule
\textbf{Method} & ATT$_{\text{unw}}$ & ATT$_{\text{wt}}$ & MaxSMD$_{\text{unw}}$ & MaxSMD$_{\text{wt}}$ & Ret. & Ret.$_{\text{wt}}$ \\
\midrule
CRM standard & -1.44 & -1.26 & 0.210 & 0.309 & 91.3\% & 89.5\% \\
CRM-CAM & -1.86 & -2.49 & 0.229 & 0.303 & 100.0\% & 100.0\% \\
PSM 1:1 & -1.01 & -0.68 & 0.085 & 0.155 & 98.7\% & 98.5\% \\
\bottomrule
\end{tabular}
\end{adjustbox}
\label{tab:nhanes_weighted}
\end{table}

\section{Minimum Reporting Standard}\label{sec:reporting}

By \Cref{thm:estimand-gap} and \Cref{cor:gap-bound}, $\hat\pi>0$ indicates
that the full treated population is not supported in the CRM grid.
We therefore recommend that matched observational analyses report six
quantities: the shortage fraction $\hat\pi$ before matching; final treated
retention $|S_M|/n_T$ after capacity constraints; the covariates
that define the unsupported subgroup $S^c$; the estimand label
(``full ATT'' only when $\hat\pi=0$ and retention is $100\%$, otherwise a
conditional ATT for the retained overlap population); any auxiliary
information about $\tau_{S^c}$; and a plain
scope statement describing the treated population to which the estimate
applies. When retention is below $1-\hat\pi$, the scope statement must also
describe capacity-driven exclusions. In the real NHANES example,
$\hat\pi=0.018$, so the support diagnostic excludes at least $1.8\%$ of
daily smokers before any additional matching attrition is considered.

\section{Discussion}\label{sec:discussion}

The empirical results support a narrow interpretation. CRM is not a
replacement for every high-quality matcher. In moderate samples, corrected
PSM and entropy balancing often achieve better marginal balance. CRM's
strength is the large-sample design regime: it avoids global pairwise
search, retains nearly all treated units when overlap is adequate, and
reports both the pre-matching shortage fraction and final retention. This is why the Criteo benchmark is the central
application, while the simulation study is used mainly to identify the
conditions under which CRM is or is not competitive.

Three limitations follow directly from the construction. First, CRM
requires representation adequacy: Assumption~\ref{ass:reprignor} is not
implied by ignorability given $X$, and nonlinear or multi-modal treatment
assignment can leave residual confounding in $Z=(d,\phi)$. Second, CRM
trades some finite-sample marginal balance for retention and scalability;
when the goal is the best possible MaxSMD in a small overlap population,
PSM or weighting may be preferable. Third, CRM's bootstrap intervals are
approximate calibration tools rather than efficiency guarantees, as
discussed below.

CRM is also compatible with outcome-modeling and doubly robust estimation.
In supplementary experiments, applying AIPW after CRM matching reduces
absolute bias by roughly a factor of four and performs similarly to a
correctly specified AIPW-OLS estimator (\Cref{fig:ipw-aipw}). We therefore view CRM as a
design-stage method that can be followed by model-based estimation, not as
a competitor to the entire class of doubly robust procedures. Cross-fitted
double/debiased machine learning estimators \citep{chernozhukov2018} are a
natural estimation-stage complement, but a full cross-fitted super-learner
benchmark is outside the matching-design comparison in the main paper.

\subsection{Uncertainty Quantification}\label{sec:uncertainty}

CRM reports paired-bootstrap uncertainty intervals as an \emph{approximate}
measure of variability (Algorithm~1, Stage~3). The paired bootstrap treats
matched pairs as i.i.d.\ and ignores uncertainty from the learned partition,
so we use it as an empirical uncertainty summary rather than a formal
post-matching variance estimator.

A calibration study across $n_T \in \{500,1{,}000,2{,}000,5{,}000,
10{,}000\}$ (100 replications each, Scenario S\ref{sc:typical}) finds
sub-nominal coverage of $0.745$--$0.870$ for the paired-bootstrap
$95\%$ interval (\Cref{fig:coverage}). The interval width shrinks with
sample size, but the nominal coverage is not guaranteed. This is consistent
with \citet{abadie2008}, who show that the nonparametric bootstrap can fail
for nonsmooth matching estimators.

To separate variance calibration from estimator bias, \Cref{app:variance}
adds an inference diagnostic with 500 replications per setting. The diagnostic
compares paired-bootstrap standard-error Wald intervals, conservative
cell-stratified Wald intervals, and outcome-model bias corrections. In three
mean-shift settings the CRM bias is small, but both bootstrap and
cell-stratified standard errors underestimate Monte Carlo variability by
roughly $20$--$25\%$. In the
covariance-driven nonlinear setting, basic CRM has large representation bias;
a linear bias correction misses it, while a quadratic correction detects and
removes it in that constructed setting. We therefore treat CRM intervals as
approximate uncertainty summaries and use the bias-correction gap as an
additional diagnostic for residual prognostic imbalance.

\paragraph{For practitioners, we recommend:}
\begin{itemize}[noitemsep]
  \item Report the paired-bootstrap interval from the software as an
        approximate uncertainty summary, not as a guaranteed nominal
        confidence interval.
  \item When CRM cells are central to the analysis, also report the
        conservative cell-stratified standard error in \Cref{app:variance};
        it is a sensitivity calculation conditional on the learned partition.
  \item Use the bias-correction gap in \Cref{app:variance} to audit whether
        the CRM representation leaves prognostic imbalance. A large gap
        indicates that the point estimate is sensitive to outcome modeling.
  \item Do \textbf{not} interpret these intervals as semiparametric
        efficiency-optimal intervals. The $O(n_T^{-1/2})$ rate result
        (Proposition~\ref{prop:rate}) concerns MSE convergence in fixed
        representation dimension, not Cramér--Rao efficiency.
\end{itemize}

Fully valid post-matching inference for CRM remains a theoretical problem:
it must account simultaneously for the learned partition, nonsmooth matching,
finite-cell sampling, and possible representation bias.

\begin{figure}[htbp]
\centering
\includegraphics[width=0.62\linewidth]{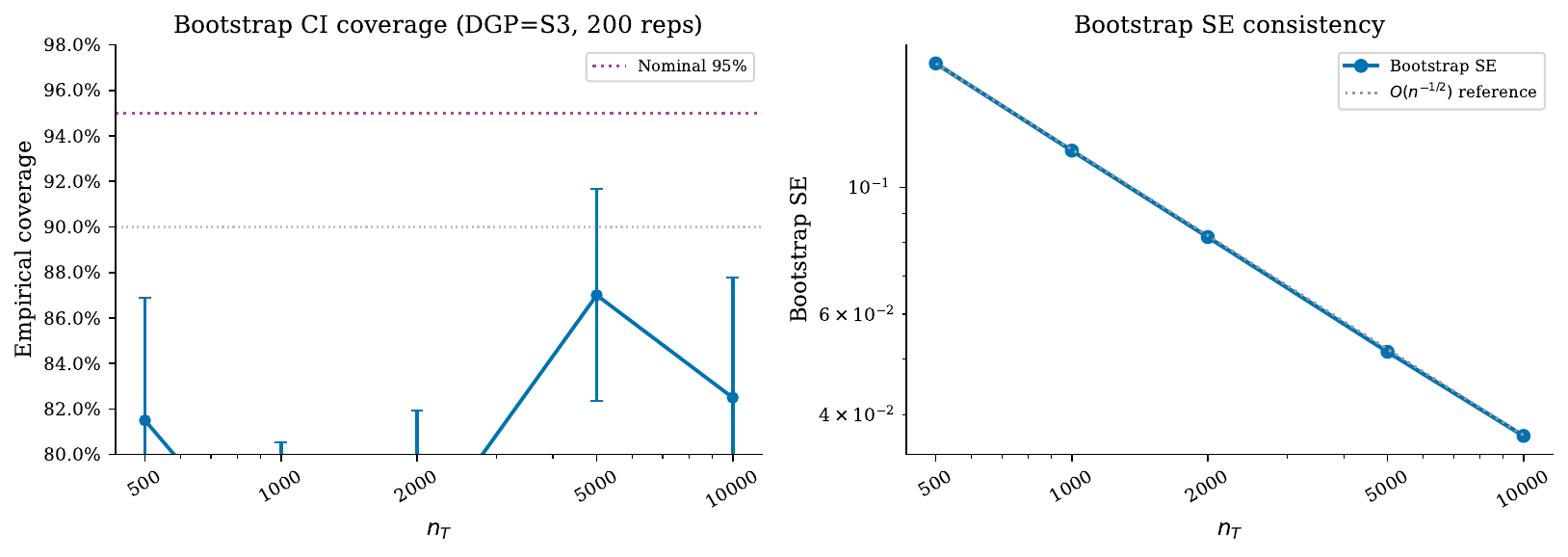}
\caption{Empirical coverage of the paired bootstrap $95\%$ interval vs.\
$n_T$ (Scenario S\ref{sc:typical}; 100 replications).
Coverage stays in the $0.745$--$0.870$ band, below nominal, while the
interval width (right axis / annotation) shrinks correctly with $n_T$.
The gap reflects the nonstandard calibration of nonsmooth matching
estimators \citep{abadie2008}; intervals are reported as empirical
calibration, not guaranteed coverage.}
\label{fig:coverage}
\end{figure}

\paragraph{Standard causal assumptions.}
CRM inherits the standard assumptions: no unmeasured confounding
(\Cref{ass:ignor}), no interference (\Cref{ass:sutva}), and consistency
of potential outcomes.

\subsection{Future Extensions}\label{sec:future}

Several extensions are natural but outside the present contribution. First,
the treated mean could be replaced by a robust reference location, such as a
trimmed centroid or geometric median, to improve stability under heavy-tailed
covariates. A multi-centroid version could similarly use treated clusters or
medoids as local reference objects, addressing multimodal treated
distributions while preserving CRM's reference-based matching structure.
Second, variance-level or interaction-based confounding could be addressed
by adding a covariance-direction coordinate to the current mean-shift Fisher
coordinate. Third, a fully design-based variance estimator could replace the
approximate bootstrap used here. Fourth, a capped CRM-NN implementation
could retain nearest-neighbor refinement in small cells while bounding
large-scale overhead.

\subsection{Software and Reproducibility}\label{sec:software}

CRM, all variants (CRM-random, CRM-NN, CRM-RunFast), and the pre-matching
diagnostic are implemented in the Python package \texttt{crm-matching}.
The public repository is available at
\url{https://github.com/KemingHu-D/crm-matching}.

A reproducibility supplement (ZIP, $\leq$100~MB) is included as supplementary
material on OpenReview. It contains all scripts, configuration files, and the
\texttt{crm} package required to reproduce every figure and table from raw
data.
The deterministic reproduction path is:
\begin{enumerate}[noitemsep]
  \item \texttt{pip install -r requirements.txt}
  \item \texttt{pytest tests/} \hfill (\textless{}2~min; all tests pass)
  \item \texttt{python experiments/run\_simulation.py --fast} \hfill (smoke-test; \textless{}5~min)
  \item \texttt{python experiments/run\_simulation.py} \hfill (full; $\approx$2~h)
\end{enumerate}
All random seeds are passed explicitly; no global state is used.
Replication $i$ uses seed $\mathtt{seed\_base} + 1000i$, matching the
paper tables exactly.

\section{Conclusion}\label{sec:conclusion}

We have introduced Centroid-Referenced Mahalanobis Matching (CRM), a
representation-based design for causal inference in large observational
studies.
By replacing pairwise unit-level matching with stratified distributional
matching over a two-dimensional geometric summary $(d,\phi)$, CRM avoids
the $O(n_Tn_C)$ pairwise-search term; its implemented cost is
$O(np^2+p^3+n\log n)$, simplifying to $O(np^2+n\log n)$ when $n\geq p$,
and it has an
$O(n_T^{-1/2})$ MSE bound in fixed representation dimension.

In three of the four simulation designs, adding the Fisher coordinate
substantially reduces the imbalance of radial-only CRM. These experiments
support $k=1$ as the default for the tested mean-shift settings, without
claiming that one direction is sufficient under arbitrary assignment.
The Criteo benchmark demonstrates two complementary findings: MaxSMD is
an insufficient proxy for estimation accuracy (CEM's near-zero MaxSMD
coexists with the highest estimation error), and CRM matches or improves on
the balance of a correctly implemented (collision-free) PSM across the
large-scale configurations while retaining at least $99.4\%$ of treated units
and running roughly an order of magnitude faster.
The LaLonde application \citep{lalonde1986} illustrates CRM-NN in a small
sample, reducing bias from $-\$908$ to $+\$135$ while retaining more
treated units than full Mahalanobis NN.
The NHANES application provides a real-data check: the pre-matching shortage
diagnostic reports only a small support restriction, while survey weighting
does not materially change the smoking--blood-pressure estimate.

As observational datasets continue to grow in scale and complexity, methods
that jointly address computational scalability, transparent estimand
definition, and transparent reporting of overlap limitations will become
increasingly important.
CRM represents a principled step in this direction.

\subsubsection*{Broader Impact Statement}
CRM is a methodological contribution to statistical causal inference.
Matching methods are widely used in policy evaluation, medical research,
and social science.
CRM's pre-matching shortage diagnostic explicitly surfaces estimand
restrictions that are implicit in all matching methods, promoting more
transparent reporting and reducing the risk of overgeneralizing causal findings.
The pre-matching diagnostic also helps practitioners identify when
additional data collection is needed rather than proceeding with a
structurally deficient analysis.
No direct negative societal impacts are anticipated from this methodological
work, beyond the standard risks that any causal inference method carries
when applied to observational data that may contain unmeasured confounding.

\subsubsection*{Acknowledgments}
We thank the TMLR action editor and reviewers for detailed comments that
substantially improved the organization, empirical evaluation, and
presentation of the paper. K.H. thanks Gregory Becker for general
professional support.

\bibliography{crm_refs}
\bibliographystyle{tmlr}

\appendix

\section{Balance--Retention Frontier}\label{app:frontier}

\Cref{fig:frontier} displays the two components of CRMSE without imposing
equal weighting, in the spirit of the balance--sample-size frontier
\citep{king2017}. In S1--S3, entropy balancing or collision-free PSM lies
closer to the ideal corner than CRM $k=1$. In S4, no method achieves both
adequate balance and high retention. This figure documents the finite-sample
tradeoff and should not be read as evidence that CRM dominates competing
methods. CRM's empirical case instead rests on its support diagnostic and its
large-sample computational behavior.

\begin{figure}[htbp]
\centering
\includegraphics[width=\linewidth]{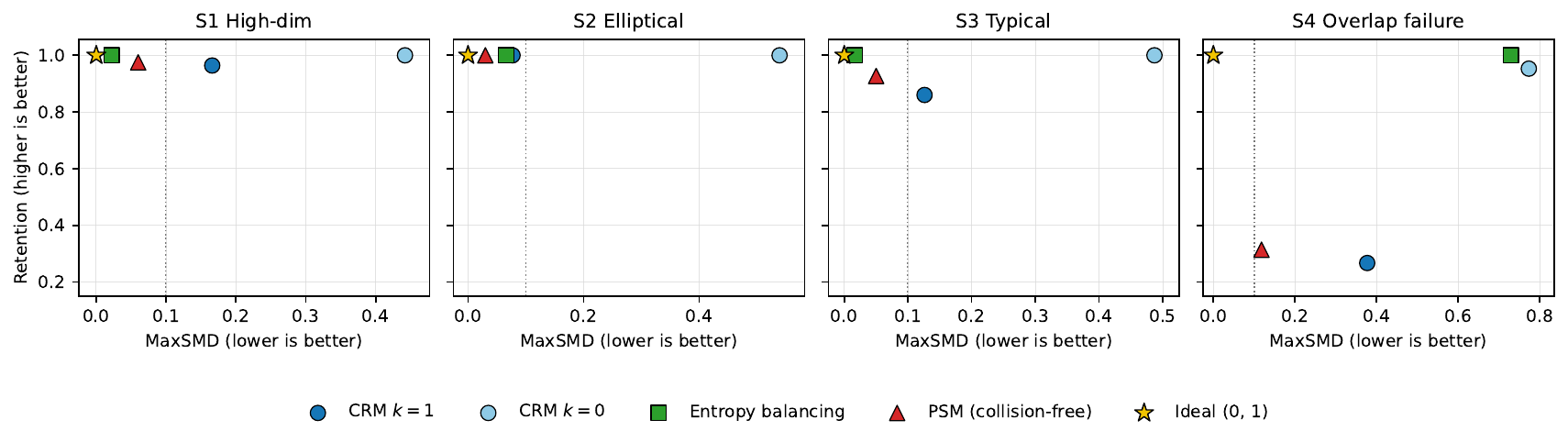}
\caption{MaxSMD and retention across the four simulation scenarios at
$n_T=1{,}000$. Each point is a method, and $(0,1)$ is the ideal of perfect
balance and full retention. The dotted line marks MaxSMD $=0.1$. The figure
separates the two components of CRMSE and shows that CRM is not uniformly
frontier-dominant in these moderate-sample settings.}
\label{fig:frontier}
\end{figure}

\section{Conservative Variance and Bias-Correction Diagnostics}
\label{app:variance}

This appendix gives the additional inference diagnostic used in
\Cref{sec:uncertainty}. It is not a new formal inferential guarantee. Its
purpose is to separate three mechanisms that can reduce coverage: variance
underestimation, bias in the CRM point estimate, and the nonsmooth form of the
matching estimator.

\paragraph{Conservative cell-stratified standard error.}
Conditional on the learned CRM partition, treat each occupied cell $k$ as a
stratum. Let
\[
\hat\tau_k=\bar Y_{T,k}-\bar Y_{C,k},\qquad
w_k=\frac{n_{T,k}}{\sum_\ell n_{T,\ell}},
\]
where $n_{T,k}$ and $n_{C,k}$ count distinct matched treated and control units
in cell $k$. The cell-stratified point estimate is
$\sum_k w_k\hat\tau_k$, and the conservative variance calculation is
\[
\widehat{\operatorname{Var}}_{\rm cell}(\hat\tau)
=
\sum_k w_k^2
\left(
\frac{s^2_{T,k}}{n_{T,k}}+
\frac{s^2_{C,k}}{n_{C,k}}
\right).
\]
For sparse cells with $n_{T,k}<2$ or $n_{C,k}<2$, the undefined within-cell
variance is replaced by the pooled within-cell variance from non-sparse cells
on the same side. This mirrors variance calculations for subclassification and
coarsened exact matching \citep{rosenbaum1984,iacus2012}, but here it is used
only as a conservative sensitivity calculation conditional on CRM's learned
cells.

\paragraph{Outcome-model bias correction.}
To diagnose residual prognostic imbalance, fit a control-outcome model
$\hat\mu_0(x)$ on the control group, following the bias-corrected matching
logic of \citet{abadie2011}. We write BC for bias correction and compute
\[
\hat\tau_{\rm BC}
=
\hat\tau_{\rm CRM}
-
\sum_k w_k
\left\{
\bar{\hat\mu}_{0,T,k}
-
\bar{\hat\mu}_{0,C,k}
\right\}.
\]
The correction term is a bias audit: a large value means that treated and
matched control units still differ in predicted untreated outcomes. We report
both a linear correction and a quadratic correction. The quadratic correction
is not part of the basic CRM matching design; it is included to test whether a
flexible outcome model can detect covariance-driven nonlinear confounding.

\paragraph{Coverage diagnostic.}
We run 500 replications in four inference-diagnostic settings, denoted I1--I4
to avoid confusion with the main simulation labels: I1 is a linear
high-dimensional mean-shift setting, I2 is a correlated elliptical mean-shift
setting, I3 is a covariance-driven nonlinear failure case, and I4 is a
multi-axis mean-shift setting. All rows in \Cref{tab:inference-diagnostic}
use Wald intervals so that the bootstrap and conservative standard-error
estimates are compared on the same scale. The software also reports percentile
paired-bootstrap intervals; neither form is claimed to have guaranteed nominal
coverage.

\begin{table}[htbp]\centering\small
\caption{Inference diagnostic coverage at nominal $95\%$ level
(500 replications per setting). ``Bootstrap-SE'' means a Wald interval using
the paired-bootstrap standard error. ``Conservative'' uses
$\widehat{\operatorname{Var}}_{\rm cell}$ above. The empirical-bias-corrected row is
simulation-only and subtracts the Monte Carlo mean bias before checking
coverage.}
\label{tab:inference-diagnostic}
\begin{tabular}{lcccc}
\toprule
Interval & I1 & I2 & I3 & I4 \\
\midrule
Paired bootstrap-SE Wald & 0.886 & 0.900 & 0.000 & 0.906 \\
Cell-stratified conservative Wald & 0.884 & 0.898 & 0.000 & 0.906 \\
Linear bias-corrected conservative Wald & 0.994 & 0.988 & 0.000 & 0.992 \\
Quadratic bias-corrected conservative Wald & 0.994 & 0.990 & 1.000 & 0.990 \\
Empirical-bias-corrected diagnostic & 0.888 & 0.898 & 0.702 & 0.900 \\
\bottomrule
\end{tabular}
\end{table}

\Cref{tab:inference-diagnostic} shows two different failure modes. In I1, I2,
and I4, CRM's average bias is close to zero, but the bootstrap and conservative
SEs are about $0.74$--$0.82$ of the Monte Carlo SD. The problem is therefore
mainly variance calibration rather than estimator centering. In I3, the raw
CRM bias is large ($-1.713$), the linear bias correction does not detect it,
and no variance estimator can restore coverage. The quadratic correction
detects almost exactly the missing prognostic structure (mean correction
$1.712$) and restores coverage in this constructed setting. Thus, the
bias-correction gap is useful as a diagnostic for representation failure, but
formal valid inference for CRM remains future work.

\paragraph{NHANES comparison.}
For the real NHANES smoking application ($n_T=668$), we report paired
bootstrap intervals as approximate empirical uncertainty summaries.
Because NHANES is a complex survey, \Cref{tab:nhanes_weighted} separately
reports survey-weighted point estimates and weighted balance; a full
design-based variance estimator with strata and primary sampling units is
outside the scope of this methodological comparison.

\section{NHANES-Calibrated Semi-Synthetic Stress Test}
\label{app:nhanes-semisynth}

The corrected real NHANES extract has only mild support shortage
($\hat\pi=0.018$), so it is not a stress test for the estimand-decomposition
diagnostic.
We therefore add a semi-synthetic experiment using the same real NHANES
covariate extract as the base, while synthetically imposing an alcohol-support
barrier and a known heterogeneous treatment effect.
This experiment is not used as evidence about smoking and blood pressure; its
purpose is to check whether the CRM diagnostic and conditional-estimand
interpretation behave as expected when the support problem is known.
\Cref{tab:nhanes_semisynth_stress} reports the result.

\begin{table}[htbp]\centering\small
\setlength{\tabcolsep}{4pt}\renewcommand{\arraystretch}{1.05}
\caption{NHANES-calibrated semi-synthetic overlap stress test.
The real NHANES covariate extract is used as the base, but alcohol support
and the outcome are synthetically generated with known truth.
The imposed shortage fraction is $\hat\pi=0.049$; the true full ATT is
$7.65$ mmHg and the true supported ATT is $7.45$ mmHg.
Bias$_{\mathrm{matched}}$ is computed relative to the true ATT of each
method's retained treated subset.}
\begin{adjustbox}{max width=\textwidth}
\begin{tabular}{lrrrrrr}
\toprule
Method & Ret. & MaxSMD & ReprSMD & ATT & True ATT$_{\mathrm{matched}}$ & Bias$_{\mathrm{matched}}$ \\
\midrule
CRM standard & 89.4\% & 0.266 & 0.174 & 7.61 & 7.53 & 0.09 \\
CRM-CAM & 99.9\% & 0.428 & 0.006 & 7.49 & 7.64 & -0.15 \\
PSM 1:1 & 94.6\% & 0.056 & 0.192 & 7.23 & 7.49 & -0.26 \\
\bottomrule
\end{tabular}
\end{adjustbox}
\label{tab:nhanes_semisynth_stress}
\end{table}

\section{Proofs}\label{app:proofs}

\subsection{Proof of Proposition~\ref{prop:rate} (Conservative Rate Bound)}

We follow the standard bias-variance decomposition for histogram regression
estimators in $d_Z = 2$ dimensions \citep{devroye1985}.

\paragraph{Bias.}
Within cell $k$ with center $z_k$ and diameter $h_n$, the cell control
mean $\hat{m}_k = |C(k)|^{-1}\sum_{j\in C(k)} Y_j(0)$ approximates
$\E[Y(0)\mid Z=z_k, T=0]$.
Under Lipschitz continuity of $m(z) = \E[Y(0)\mid Z=z]$ with constant
$L$, the within-cell approximation satisfies
$\bigl|\E[\hat{m}_k] - m(z_k)\bigr| \leq L\,h_n + |\Delta(z_k)|$.
Under \Cref{ass:reprignor}, $\Delta \equiv 0$, giving
$\mathrm{Bias}^2 = O(h_n^2)$.

\paragraph{Variance.}
Each cell contains on average $n_C h_n^{d_Z} f_Z(z_k)$ control units,
where $f_Z$ is the density of $Z(X)$.
Hence $\mathrm{Var}(\hat{m}_k) = O(1/(n_C h_n^{d_Z}))$.
The ATT estimator averages $\taucrm = |S|^{-1}\sum_{i\in S}(Y_i - \hat{m}_{k(i)})$.
The treated outcomes $Y_i(1)$ contribute $O(1/n_T)$ to variance, and
bounding the variance of the weighted average by the largest cell-mean
variance gives the conservative order
$O(1/(n_T h_n^{d_Z}))$ when $n_C\propto n_T$.
This bound does not exploit the additional variance reduction from averaging
across cells and is therefore not asserted to be sharp.

\paragraph{Balancing the bound.}
Setting $\mathrm{Bias}^2 = \mathrm{Variance}$:
$h_n^2 = 1/(n_T h_n^{d_Z})$, so $h_n^{2+d_Z} = 1/n_T$, giving
\[
h_n^* = n_T^{-1/(2+d_Z)}.
\]
Substituting into MSE $= O(h_n^{*2})$:
\[
\mathrm{MSE}(\taucrm) = O\!\left(n_T^{-2/(2+d_Z)}\right).
\]
For $d_Z = 2$: $\mathrm{MSE} = O(n_T^{-1/2})$.
For $d_Z = p$ (full-dimensional matching): $\mathrm{MSE} = O(n_T^{-2/(2+p)})$.
The ratio is $n_T^{(p-2)/(2(2+p))}$, which diverges for $p>2$. \qed

\section{PCA-Matching Sensitivity}\label{app:pca-sens}

\Cref{fig:pca-sens} reports the sensitivity of PCA-95\%+NN balance to the
variance-retention threshold and the number of retained components, the
comparison referenced in \Cref{sec:pca-empirical}.
Unlike CRM's default two-coordinate representation, PCA-matching requires a
variance-retention threshold and its balance varies
non-monotonically with that choice: retaining too few components discards
directions relevant to treatment or outcome, while retaining too many reintroduces the
high-dimensional matching problem CRM is designed to avoid.

\begin{figure}[htbp]
\centering
\includegraphics[width=\linewidth]{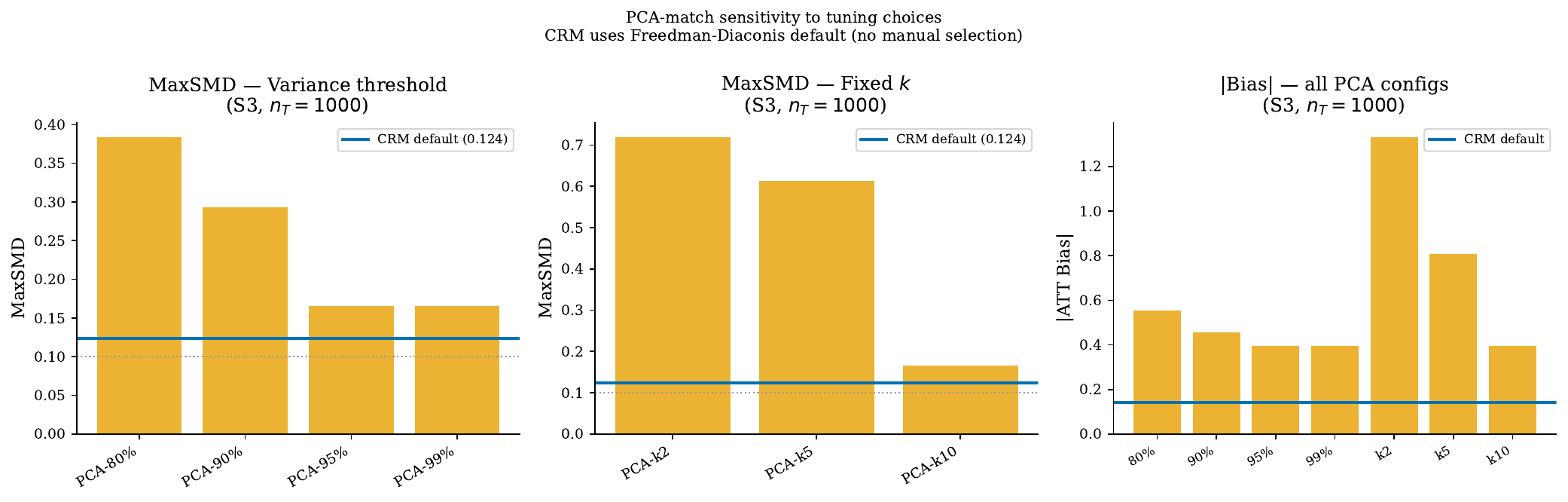}
\caption{Balance of PCA-95\%+NN as a function of the retained
variance fraction / number of principal components.
MaxSMD varies non-monotonically with the threshold, with no stable
``correct'' setting across data-generating processes; CRM has no analogous
tuning knob.}
\label{fig:pca-sens}
\end{figure}

\section{CRM-NN Timing Supplement}\label{app:crm-nn}

\begin{figure}[htbp]
\centering
\includegraphics[width=\linewidth]{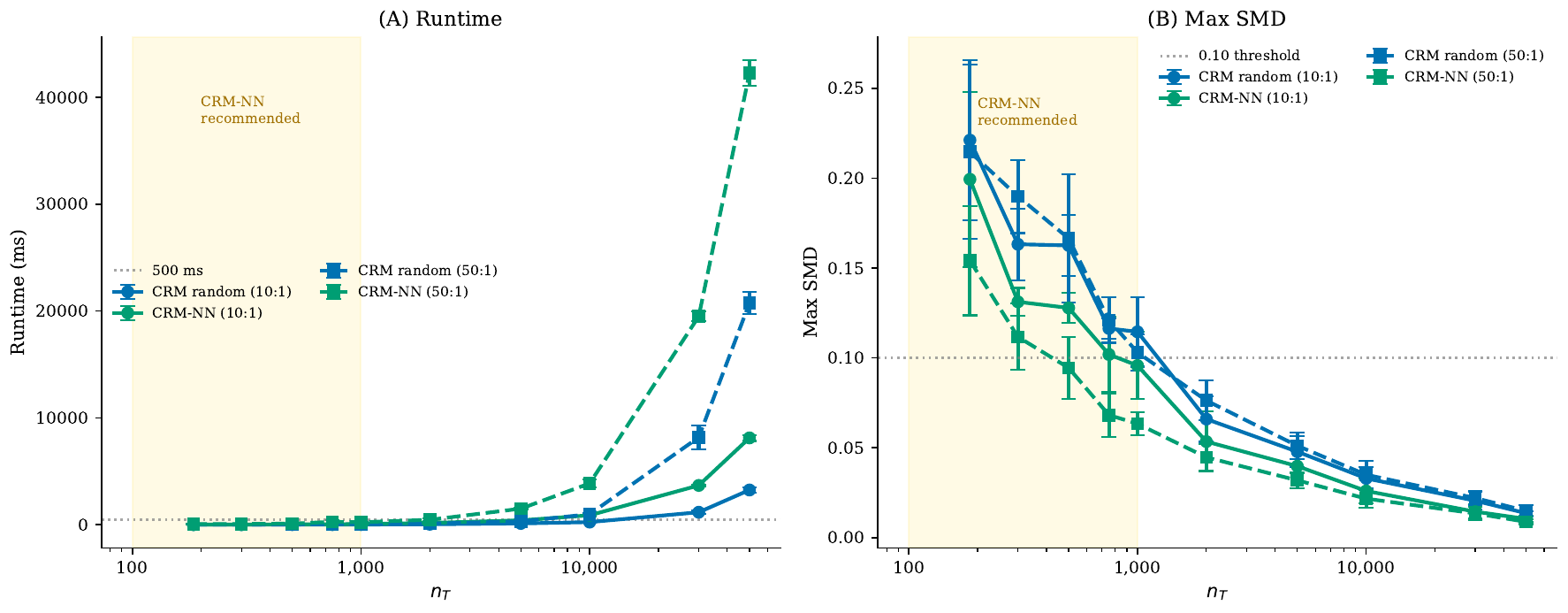}
\caption{CRM-random vs.\ CRM-NN across sample sizes
($p=8$, Scenario A, $5$ data seeds $\times$ $7$ timing replicates;
gold shading marks the CRM-NN recommended regime $n_T<1{,}000$).
\textbf{(A)}~Runtime: NN overhead stays below 500~ms for
$n_T\leq1{,}000$ at both tested control ratios, then grows
proportionally to $n_C/n_T$.
\textbf{(B)}~Max SMD: the quality gap persists at all $n_T$ because
average controls per cell grows as
$\bar{n}_C^{\mathrm{cell}} \propto n_T^{1/3}$, maintaining a pool
within each cell for NN to improve upon.}
\label{fig:crm-nn}
\end{figure}

\section{Decision Flowchart}\label{app:flowchart}

\begin{center}
\noindent\begin{minipage}{0.94\linewidth}
\small\setlength{\fboxsep}{10pt}
\fbox{\begin{minipage}{0.92\linewidth}
\textbf{Method Selection Flowchart}\\[0.5em]
\textbf{Step 1. Sample size.}
\begin{itemize}
  \item $n_T < 200$ $\to$ Use 1:1 Mahalanobis NN (cells too sparse for binning).
  \item $200 \leq n_T < 1{,}000$ $\to$ Use CRM-NN (\Cref{sec:wcnn}).
  \item $n_T \geq 1{,}000$ $\to$ Use CRM-random (\Cref{sec:estimator}).
\end{itemize}
\textbf{Step 2. Shortage fraction.}\\
Compute $\hat\pi$ via \Cref{eq:pihat} (before any matching).
\begin{itemize}
  \item $\hat\pi < 5\%$ $\to$ Full ATT approximately estimable.
  \item $5\% \leq \hat\pi < 30\%$ $\to$ Conditional ATT; report scope statement.
  \item $\hat\pi \geq 30\%$ $\to$ Severe overlap failure; expand control pool
        or restrict estimand explicitly before proceeding.
\end{itemize}
\textbf{Step 3. Confounding structure.}
\begin{itemize}
  \item Mixed continuous/binary covariates $\to$ CRM-CAM (\Cref{sec:variants}).
  \item Two independent confounding axes suspected $\to$ CRM $k=2$ (future work).
  \item Default $\to$ CRM $k=1$ (\Cref{sec:repr}).
\end{itemize}
\end{minipage}}
\end{minipage}
\end{center}

\section{CRM Variants: Full Descriptions}\label{app:variants}

\paragraph{CRM-CAM (Correlation-Adjusted Mahalanobis).}
Replaces $\SigT$ with the correlation-scale matrix
$R = D^{-1/2}\SigT D^{-1/2} + \varepsilon I_p$, where $D = \diag(\SigT)$,
placing all covariates on a unit-variance scale.
Recommended for mixed continuous/binary covariate sets where raw variance
differences would otherwise dominate the Mahalanobis geometry.

\paragraph{CRM-Pool (Pooled Bin Edges).}
Uses bin edges from the pooled $(d,\phi)$ distribution of both groups,
reducing treated attrition when the centroid shift is large and treated
units concentrate near the edge of the control support.

\paragraph{CRM-Trim (Explicit Support Trimming).}
Formally discards treated units with $d(X_i)$ outside the central $90\%$
of the control distance distribution before matching, and reports them
separately as an explicit overlap restriction with the associated estimand
caveat.

\section{Geometric Supporting Results}\label{app:geometry}

\begin{proposition}[Chi distribution of $d$]\label{prop:chi}
Let $X \mid T=1 \sim \mathcal{N}(\mu_T, \Sigma_T)$ with $\Sigma_T$
positive definite.
Then $d^2(X)\mid T=1 \sim \chi^2_p$ with $\E[d^2]=p$,
$\mathrm{Var}(d^2)=2p$, and
$\E[d(X)\mid T=1] = \sqrt{2}\,\Gamma((p+1)/2)/\Gamma(p/2)$.
\end{proposition}

\begin{proof}
$Z = \hat{L}^{-1}(X-\mu_T) \sim \mathcal{N}(0,I_p)$, so
$d^2(X) = \|Z\|_2^2 \sim \chi^2_p$.
\end{proof}

\begin{proposition}[Independence of distance and direction]\label{prop:indep}
Under $X \sim \mathcal{N}(\mu_T, \Sigma_T)$, $d(X)$ and
$u(X) = \hat{L}^{-1}(X-\mu_T)/d(X)$ are statistically independent,
with $u(X)$ uniform on $\mathcal{S}^{p-1}$.
\end{proposition}

\begin{proof}
Writing $Z = d\cdot u$ in polar coordinates, the Jacobian is $d^{p-1}$,
giving $f_{d,u}(r,\omega) = (2\pi)^{-p/2}\exp(-r^2/2)\cdot r^{p-1}$,
which factors as $f_d(r)\cdot f_u(\omega)$.
\end{proof}

\begin{proposition}[Complexity]\label{prop:complexity}
Algorithm~1 has implemented time complexity $O(np^2+p^3+n\log n)$ and space
complexity $O(p^2+n)$, where $n=n_T+n_C$.
The term $O(np^2)$ is the dominant linear-algebra cost; the
$O(n\log n)$ term comes from quantile binning and sorting.
The bound contains no $n_Tn_C$ pairwise-search term.
Brute-force nearest-neighbor matching in $p$ dimensions requires
$O(n_Tn_Cp)$ distance work, although scalar propensity-score matching can
use sorting and need not incur this bound.
\end{proposition}

\begin{proof}
Dominant steps: covariance estimation costs $O(n_Tp^2)$, Cholesky
factorization costs $O(p^3)$, and computing distances and projections for
all units costs $O(np^2)$.
Sorting or quantile computation for binning costs $O(n\log n)$ in the
standard implementation.
Thus total implemented time is $O(np^2+p^3+n\log n)$ and space is
$O(p^2+n)$. For the regimes considered here, $n\gg p$, so $p^3$ is
absorbed by $np^2$.
\end{proof}

\section{Relationship to Doubly Robust Estimation}
\label{app:dr-scope}

Doubly robust (DR) estimators such as augmented inverse probability weighting
(AIPW), targeted learning, and cross-fitted double/debiased machine learning
\citep{bang2005,chernozhukov2018}
address a different, estimation-stage question from matching design. They
can target full or restricted populations depending on the overlap and
weighting choices. The main experiments therefore compare matching designs
and report balance and retention, while this appendix provides a focused
IPW/AIPW sensitivity analysis based on estimation error. It is not an
exhaustive comparison of outcome learners or cross-fitting schemes.
Entropy balancing \citep{hainmueller2012} remains in the main simulation as
a scalable representative of balancing weights.
\Cref{fig:ipw-aipw} reports the resulting sensitivity comparison.

\begin{figure}[htbp]
\centering
\includegraphics[width=0.72\linewidth]{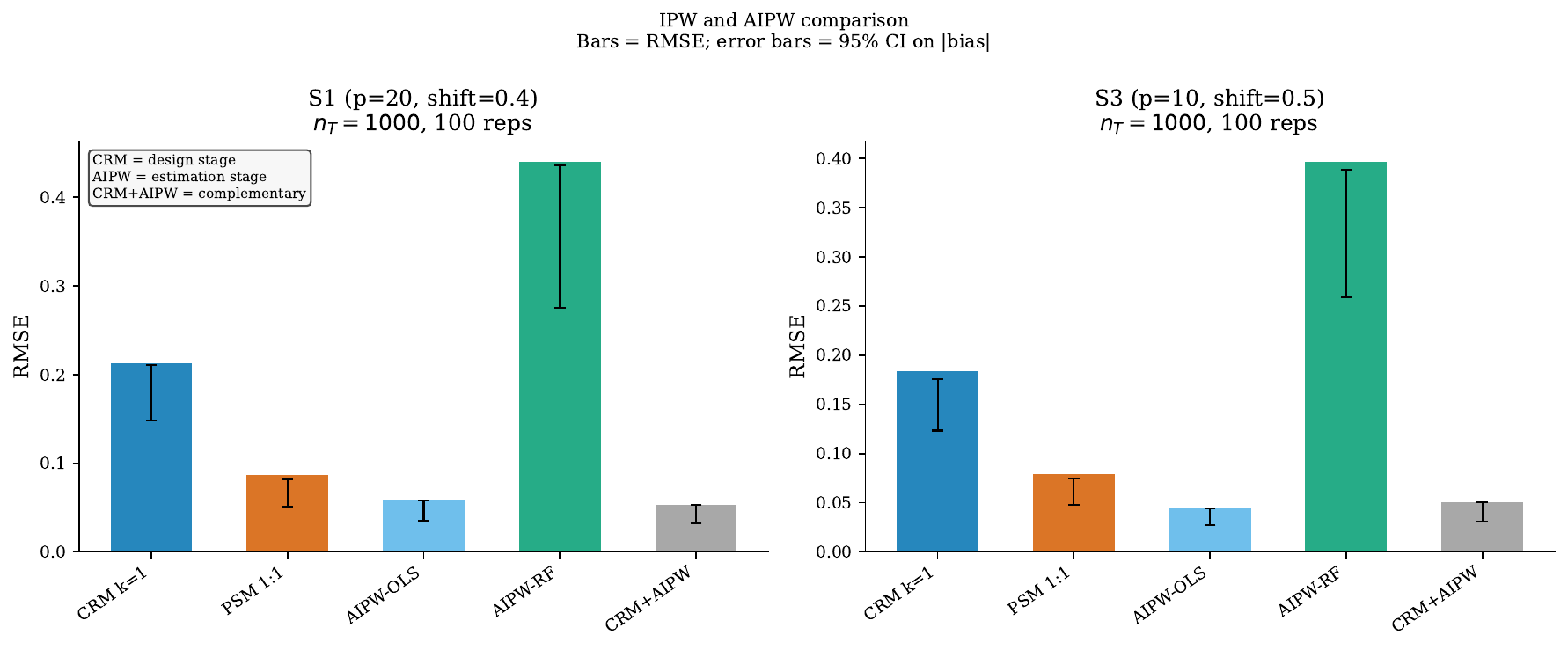}
\caption{Supplementary comparison with inverse probability weighting (IPW),
AIPW, and matching-plus-AIPW estimators. The final estimator can incorporate
CRM as a design-stage restriction before AIPW-style estimation.}
\label{fig:ipw-aipw}
\end{figure}

\section{When CRM May Underperform}\label{app:failure}

\begin{enumerate}
\item \textbf{High balance requirement at moderate $n_T$.}
In the moderate-size synthetic sweeps, PSM often achieves lower MaxSMD than
CRM (gap $0.006$--$0.133$ in the reported comparison).
If per-covariate MaxSMD below $0.05$ is the primary criterion and retention
loss is acceptable, PSM is preferable.

\item \textbf{Nonlinear confounding.}
The linear Fisher direction $v$ cannot capture interaction-based or
higher-order confounding.
PSM with a flexible treatment model or MALTS \citep{parikh2022} may achieve
lower bias in such settings.

\item \textbf{Very small $n_T$ (below 200).}
Even CRM-NN is limited when the average cell contains fewer than $2$--$3$
treated units.
Full 1:1 Mahalanobis nearest-neighbor matching is preferable in this regime.
\end{enumerate}

\end{document}